\documentclass{ecai} 

\usepackage{latexsym}
\usepackage{amssymb}
\usepackage{amsmath}
\usepackage{amsthm}
\usepackage{booktabs}
\usepackage{enumitem}
\usepackage{graphicx}
\usepackage{color}
\usepackage{dsfont}

\usepackage{bm}
\usepackage{algorithm}
\usepackage{algpseudocode}
\usepackage{caption}
\usepackage{subcaption}
\usepackage{hyperref} 

\newtheorem{theorem}{Theorem}
\newtheorem{lemma}[theorem]{Lemma}

\newtheorem{definition}{Definition}

\newcommand{\BibTeX}{B\kern-.05em{\sc i\kern-.025em b}\kern-.08em\TeX}

\begin{document}


\begin{frontmatter}

\paperid{2635}

\title{Learning in Repeated Multi-Objective Stackelberg Games with Payoff Manipulation}
\author{\fnms{Phurinut}~\snm{Srisawad}}
\author{\fnms{Juergen}~\snm{Branke}}
\author{\fnms{Long}~\snm{Tran-Thanh}} 

\address{University of Warwick, United Kingdom}

\begin{abstract} We study payoff manipulation in repeated multi-objective Stackelberg games, where a leader may strategically influence a follower’s deterministic best response, e.g., by offering a share of their own payoff. We assume that the follower’s utility function, representing preferences over multiple objectives, is unknown but linear, and its weight parameter must be inferred through interaction. This introduces a sequential decision-making challenge for the leader, who must balance preference elicitation with immediate utility maximisation. We formalise this problem and propose manipulation policies based on expected utility (EU) and long-term expected utility (longEU), which guide the leader in selecting actions and offering incentives that trade off short-term gains with long-term impact. We prove that under infinite repeated interactions, longEU converges to the optimal manipulation. Empirical results across benchmark environments demonstrate that our approach improves cumulative leader utility while promoting mutually beneficial outcomes, all without requiring explicit negotiation or prior knowledge of the follower’s utility function.
\end{abstract}

\end{frontmatter}


\section{Introduction}

Stackelberg games provide a foundational framework for hierarchical decision-making, where a leader commits to a strategy first and a follower responds optimally. Traditionally, Stackelberg games have been employed to model a variety of competitive environments, where the leader seeks to maximise their own utility while anticipating the follower’s best response \citep {labbe1998bilevel,yu2017incentive,sinha2018stackelberg}. This first-mover advantage allows the leader to strategically shape the boundaries of the follower’s decision space through their own actions.

However, many real-world decision-making problems are inherently multi-objective \citep{ruadulescu2024world,godhrawala2023dynamic,brown2012multi}, including traffic management, taxation policy design, resource allocation, security, goods exchange, etc. 
Agents may have heterogeneous preferences over multiple objectives which leads to the possibility of mutually beneficial exchanges. For instance, a leader may prefer apples over bananas, while the follower favours bananas over apples. Given a leader's action, the follower selects their best-response strategy to maximise banana acquisition, which might inadvertently result in the leader receiving fewer apples. However, if the leader can influence the follower by, for example, offering a portion of their own payoff, they can steer the follower toward actions that also benefit the leader, resulting in a mutually beneficial outcome. 

 In the absence of explicit negotiations, this paper explores the concept of \textit{payoff manipulation}, where the manipulator offers an incentive, e.g., a fraction of their own payoff to influence the other players’ decision-making. With multi-objective payoffs, the manipulator must know an opponent's utility function to create a strictly preferable action. However, an opponent's utility function is typically unknown, and the manipulator needs to learn it. By observing how the follower reacts to various incentives over multiple interactions, the leader can gain insights into the follower’s preferences and adapt their strategy accordingly. This introduces an additional strategic layer, where the leader must optimise not only their immediate utility but also their ability to shape the follower’s behaviour over time.

The payoff manipulation or incentivisation has been studied in various contexts in game theory, economics, and bandit literature. \citep{eidenbenz2007manipulation} studies mechanism design where the payoff matrix of a game can be influenced by an interested third party. \citep{gupta2015pays} studies 
incentive equilibria in bi-matrix games. \citep{mahesh2023playing} introduces the use of payoff matrix manipulation by one player in the context of repeated polymatrix games. \cite{scheid2024incentivized} studies a learning problem where the principal incentivises the agent in a bandit setting.

Our work extends this concept to a multi-objective setting, where utility signifies preference levels. The contribution of this paper is to formulate a learning problem in repeated multi-objective Stackelberg games via payoff manipulation, focusing on how the leader navigates the trade-off between learning the follower’s preferences and maximising their immediate utility.  Assuming a linear utility function for the follower with unknown parameters, we propose the EU and the longEU policy to maximise myopic and non-myopic expected utility, respectively, and derive a sequential manipulation strategy to incentivise the follower. We prove that if the number of game rounds is infinite, the longEU policy will eventually converge to the optimal manipulation.  We empirically test our policies over several benchmarks to observe the proposed policy's behaviour and investigate its performance in terms of cumulative regret.  

\section{Problem Formulation: Multi-objective Stackelberg Game with Leader's Manipulation}
\label{sec:problem}
We consider a multi-objective Stackelberg game with two players, Leader (L) and Follower (F), and finite discrete action spaces $\mathcal{A}^L,\mathcal{A}^F$, respectively. At each round $t$, the leader first chooses the action $l^t \in \mathcal{A}^L$, then the follower observes the leader's action and decides their action $f^t \in \mathcal{A}^F$. Denote $X(l,f)\in \mathbb{R}^D$ as the payoff vector from the payoff matrix $X$ at the leader's action (row) $l$ and the follower's action (column) $f$. At the end of round $t$, they independently receive a payoff, a multi-dimensional vector $X^L(l^t,f^t),X^F(l^t,f^t) $, respectively. Each player has a different real-valued utility function $u^L(\bm{x})$ and $u^F(\bm{x})$, respectively, to measure their preference on a payoff vector $\bm{x}$. The follower is assumed to be a \textit{fully rational} and myopic agent who only chooses the best response (BR) action based on a leader's action $l$ and their utility function; $BR(l)=\underset{f \in \mathcal{A}^F}{\mathrm{argmax}} \ u^F(X^F(l,f))$. Therefore, we can restrict a leader's policy to a pure strategy since the follower knows the leader's action in advance. 

In this game, the leader can incentivise or manipulate the follower to choose a different action (the manipulated action) by offering the follower a payoff $\bm{c} \in \mathcal{C} \subseteq \mathbb{R}^D$ (or spending a cost $\bm{c}$) where $ \mathcal{C}$ is the space of possible manipulation costs. The follower will change their original BR if their utility of the manipulated action is higher than or equal to their best response without manipulation. The leader's goal is to maximise their overall utility, which may involve paying a cost for manipulation to incentivise the follower to respond with a favourable (for both the leader and the follower) action.
Mathematically, the leader wants to solve the following optimal manipulation problem (OMP),
\begin{equation} \label{eq:optimal_policy}
\begin{array}{llll} 
   \tag{OMP}  &\text{max}_{\bm{c},l,f} & u^L(X^L(l,f)-\bm{c}) \\
     s.t. & \text{(IC)}  & u^F(X^F(l,f)+\bm{c}) \geq u^F(X^F(l,BR(l))) \\
     & \text{(C1)} & c_d  \geq 0, \forall d \in [D]   \\
     \text{or} & \text{(C2)} & 0 \leq c_d \leq X^L_d(l,f), \forall d \in [D]  
\end{array} 
\end{equation} 
The incentive constraint (IC) is the requirement to change the follower's BR to the manipulated action $f$. The additional cost constraints (C1) and (C2) depend on the characteristics of the problem. 
\begin{itemize}
    \item (C1): Allow only positive costs similar to bribing the follower. This can make the leader's payoff negative in some objectives. One interpretation is that the leader has access to additional resources for manipulation.  
    \item (C2): Similar to (C1) but this constraint doesn't allow to spend more than what the leader will receive.
\end{itemize}

One can consider a parametric utility function denoted by $u(\bm{x};\bm{w})$ where $\bm{x}$ is a payoff and $\bm{w}$ is a parameter (or weight) from a compact space $\mathcal{W}$. In this paper, we assume linear utility functions for both leader and follower, i.e., $u(\bm{x};\bm{w})=\bm{w}\cdot\bm{x}=\sum_{d=1}^D{w_dx_d}$ where the sum of normalised weights in all objectives is $\sum_{d=1}^D{w_d}=1, \bm{w} \geq 0$. Some results on a non-linear Cobb-Douglas utility function can be found in the Appendix~\ref{sec:appendix_results}. The leader is assumed to know the follower's payoff matrix $X^F$, their own matrix $X^L$ and their own utility function parameters, so they need to estimate the follower's utility function from the history of observed follower's actions given the leader's actions. 
In every round $t$, the leader chooses an action $l^t$ and spends a manipulation cost (i.e., offers an incentive) $\bm{c}^t$. The follower chooses an action greedily, and the chosen action is observed by the leader.
Given the follower's action and the known payoff matrix, the leader can deduce a pairwise comparison feedback on the follower's utility: The payoff of the action picked is preferred over the payoff of all other actions. The set of $N$ pairwise comparisons data in the history is denoted by $\mathcal{D}:=\{(\bm{x}_i,\bm{y}_i,z_i)\}_{i=1}^N$ where $z_i=1$ if the follower prefers payoff $\bm{x}_i$ over $\bm{y}_i$; and 0 otherwise. Note that the term "manipulation" usually refers to the triple of choices (leader's action, manipulated action, manipulation cost) or $(l,f,\bm{c})$; nevertheless, we also use "manipulation" to refer to the manipulation cost only if the context is clear.

\section{Solution Concept}
  The goal for our setting is slightly different from common Stackelberg settings, which aim to quickly converge to a Stackelberg equilibrium, since the leader needs to trade off between having better knowledge of the follower's utility function and maximising their utility value. The goal for the leader, therefore, becomes minimising cumulative regret up to round $T$ defined by:
\begin{align*}
    \text{CR(T)}=\sum_{t=1}^T \left[ \left(u^L(X^L(l^*,f^*)-\bm{c}^* \right) - \left(u^L(X^L(l^t,f^t)-\bm{c}^t \right) \right]
\end{align*}
where the optimal manipulation is leader's action $l^*$, manipulated action $f^*$ and cost $\bm{c}^*$ that solve the maximization problem \eqref{eq:optimal_policy}.

\begin{lemma} \label{thm:optimal_cost_C1}
Under constraint (C1) and linear utility, there is an optimal manipulation cost $\bm{c}^*$ from the solution of \eqref{eq:optimal_policy} which is either
\begin{itemize}
    \item [1)] $\bm{c}^*=\bm{0}$ (no manipulation) or
    \item [2)] $\bm{c}^*=(0,...,0,c_d,0,...,0)$ where $c_d>0$ for some $d \in [D]$ (spending only one objective)
\end{itemize}
\end{lemma}

\begin{lemma} \label{thm:optimal_cost_C2}
Under constraint (C2) and linear utility, there is an optimal cost $\bm{c}^*$ from the solution of \eqref{eq:optimal_policy} which is either
\begin{itemize}
    \item [1)] $\bm{c}^*=\bm{0}$ (no manipulation) or
    \item [2)] $\bm{c}^*=(\bm{0},c_d,\bm{c}_{max})$ where $c_d\geq 0$  and $\bm{c}_{max}=(X^L_{d_1},X^L_{d_2},...,X^L_{d_k})$ for some $d,d_1,d_2,...,d_k \in [D]$ \\ (spending the maximum payoff from some objectives and cover the rest from one objective)
\end{itemize}
\end{lemma}

Generally, the leader cannot always benefit from a manipulation, it depends on the game and the participants' utilities. 
\begin{definition}
    A non-beneficial game is a game whose  optimal cost of the \eqref{eq:optimal_policy} is a zero-vector cost. Otherwise, it is called a beneficial game.
\end{definition}

\section{Related Work}
There are two key areas of related works: Multi-objective normal form games and learning in repeated Stackelberg games. The goal of learning the follower's preferences also links our work to preference elicitation, and we refer the interested reader for further context to \citep{qian2015learning}. In the following, we will discuss each in turn.

  \paragraph{MONFGs} are a generalisation of (single-objective) normal-form games to vectorial payoffs.  To practically deal with vectorial rewards, \citep{roijers2013survey} consider multiple conflicting objectives, the agents should balance these in such a way that the user utility derived from the outcome of a decision problem is maximised. Taking such a utility-based approach, there are two main optimisation criteria focusing on the expected scalarised returns (ESR) and the scalarised expected returns (SER) \citep{ruadulescu2020multi}. The SER criterion computes the expected value of the payoffs of a joint strategy first and then applies the utility function.  Alternatively, the ESR criterion applies the utility function before computing the expectation.  Which of these criteria should be considered best depends on how the games are used in practice. However, ESR and SER criteria are equivalent in cases where linear utility functions are used \citep{ruadulescu2020utility}. Therfore, in our setting the minimising-regret goal is equivalent to maximising a cumulative utility under SER and ESR criteria.  One recent work \citep{ropke2022preference} introduces the idea of preference communication between two players inspired by Stackelberg games under the SER criterion. Instead of manipulating the payoff matrix, players learn policies by repeatedly playing a base MONFG and following the preference communication protocol, which alternates between them being a leader who communicates a certain preference and a follower who responds to this communication. Some protocols can accelerate convergence. Additionally, some protocols may result in a cyclic Nash equilibrium. For readers interested in broader perspectives, we refer to the recent survey by \citep{ruadulescu2020multi} on multi-objective decision-making in multi-agent systems (MAS).

 \paragraph{Learning in Repeated Stackelberg Games} is another relevant field for our study. Here, we will provide a non-exhaustive overview of recent settings on this topic. \cite{sessa2020learning} assumes there are different parameterised follower types that can be chosen by an adversary. The leader knows the reward function but needs to learn the unknown response function of the follower while minimising the cumulative regret. \citep{zhao2023online} proposes a cooperative Stackelberg game with an omniscient follower where both players receive a noisy reward from a shared reward function and only the follower has full knowledge of such a reward function. The leader’s objective is to minimise the cumulative regret given the best response of the follower. \citep{assos2024maximizing} considers the follower side where the follower knows the leader’s payoff matrix and the leader’s
online learning algorithm. The follower maximises its own reward while taking into account the leader’s behaviour. A very recent work \citep{maddux2025no} studies single-leader multi-follower Stackelberg games.
Assuming that 
followers aim to achieve a Nash equilibrium given a leader's action,  the leader seeks to reach a Stackelberg equilibrium.
While most papers assume that the follower is myopic (i.e., follows a best-response dynamic), ~\cite{haghtalab2022learning} deviates from this assumption and investigates the case of non-myopic followers. Closest to our work is perhaps~\cite{haghtalab2023unifying} where learning in multi-objective Stackelberg games is studied. To the best of our knowledge, our paper is the first to allow the use of manipulation in a multi-objective Stackelberg game, which can be used to learn the follower's preferences. As such, these existing works do not need to deal with an extra level of complexity in finding the optimal manipulation strategy as we do.

In a related model from economics and bandits literature, a repeated Principal-Agent model is an incentive design problem between the principal (leader) and the agent (follower), where the aim of the principal is to minimise cumulative regret via payment (utility) transfer, which can be seen as a manipulation of the payoff matrix in case of a single objective \cite{dogan2023repeated,ho2014adaptive,scheid2024incentivized,liu2025learning}. 

\section{Manipulation Policy}
With the assumption that the follower's action is fully rational, the leader can deterministically reduce the feasible region of weight space for the follower's utility function. Given the pairwise comparison data $\mathcal{D}$, let us define the feasible region by
$ FR(\bm{w};\mathcal{D}):=\{\bm{w}\in \mathcal{W};  \text{ s.t. }   \mathds{1}\left[u^F(\bm{x}; \bm{w} ) \geq u^F(\bm{y};\bm{w})\right]=z \  \text{for all} \ (\bm{x},\bm{y},z) \in \mathcal{D} \} $.

\subsection{Initialisation and Prediction of the Best-Response} \label{sec:init_pred_BR}
To meaningfully manipulate the payoff matrix, the leader has to know the BR for each of the leader's actions first. Naively, the leader can simply test out all of their actions to investigate all follower BRs. However, this approach is inefficient if the number of leader's actions is large, since the leader may receive many low-utility payoffs. Instead, the leader can use a prediction of the follower's BR by calculating, for each leader's action $l$ and possible response $f$, the probability of being the BR conditional on the feasible region, i.e., 
\begin{align}
    &\mathbb{P}\left[BR(l)=f | l,W\right]:=  \\ 
    & \int \prod_{r  \in \mathcal{A}^F\backslash\{f\}} \mathds{1}\left[  u^F(X^F(l,f); W) \geq u^F(X^F(l,r);W) \right] dW. \notag 
\end{align}
where $W$ is a random variable representing the follower's utility weight over the feasible region $FR(\bm{w};\mathcal{D})$.

Since we do not assume to have prior information on the follower's weight, we simply assume $W$ follows the uniform distribution over the feasible region (non-informative prior). We predict the BR by $\hat{BR}(l)=\underset{f \in \mathcal{A}^F}{\mathrm{argmax}} \ \mathbb{P}\left[BR(l)=f | l,W \right]$.  Given the predicted BR, the leader needs to decide the manipulated action and their offer. For action $l$, offering a payoff dominating the BR's payoff $X^F(l,BR(l))$ (i.e., being better in all objectives) provides no information about the follower's weight since the follower definitely accepts the offer. To shrink the boundary of the feasible weight region effectively, the leader has to spend an \textit{informative} cost.
\begin{definition}
    A cost $\bm{c}$ is \emph{informative} if it makes at least one objective of the manipulated payoff $X^F(l,f)+\bm{c}$ bigger than the BR's payoff $X^F(l,BR(l))$  and makes at least one objective smaller.
\end{definition}
However, since the goal is not to find the actual weight but to maximise cumulative utility, the manipulation policy should consider the expected utility as well as the benefit of further restricting the feasible weight region.

\subsection{Expected Utility (EU) and Long-term Expected Utility (longEU) Policies}
One possible manipulation policy may be based on the highest expected utility (EU), which is defined by 
\begin{align} \label{eq:EU}
    \text{EU}( l,f,\bm{c} \ | \ W) & = \underbrace{u^L \left(X^L(l,f)-\bm{c} \right)}_{\text{new utility if accepted}} \cdot \mathbb{P} \left[ \text{accept} | l,f,\bm{c} ,W \right] \\
    & + \underbrace{u^L \left(X^L(l,BR(l)) \right)}_{\text{utility given BR if failed}}\cdot (1-\mathbb{P} \left[ \text{accept} | l,f,\bm{c} ,W \right]), \notag
\end{align}

\noindent where $\mathbb{P} \left[ \text{accept} | l,f,\bm{c} ,W \right]$ refers to the probability that the follower accepts the offer  $\bm{c}$ and chooses action $f$ given a probability distribution of the random weight $W$.

To approximate the benefit of learning about the follower's utility function, we  assume that after this round, the leader will choose the current best cost so far to manipulate the follower, and extrapolate until the end of the time horizon. In other words, we need a cost that maximises the long-term expected utility at the current round $t_0$ up to round $T$, which is calculated by
\begin{align} \label{eq:longEU}
    & \text{longEU}( l,f,\bm{c} \ | \ t_0,T,W) =  \\
    & u^L \left(X^L(l,f)-\bm{c} \right)  \cdot (T-t_0+1) \cdot \mathbb{P} \left[ \text{accept} | l,f,\bm{c} ,W \right] \notag \\ 
    &+\left( \left(u^L (X^L(l,BR(l)\right) +  u^L \left(X^L(\hat{l},\hat{f})-\hat{\bm{c}} \right) \cdot (T-t_0) \right) \notag \\ 
    & \cdot (1-\mathbb{P} \left[ \text{accept} | l,f,\bm{c} ,W \right]), \notag
\end{align}
where $(\hat{l},\hat{f},\hat{\bm{c}})$ is the current best manipulation at round $t_0$. 
Intuitively, the manipulation choice of the longEU policy is less conservative than the EU policy, i.e., the longEU policy tends to explore a new cost instead of exploiting the current best manipulation $(\hat{l},\hat{f},\hat{\bm{c}})$, especially in early rounds, since the loss from the follower refusing the offer is small compared to the potential long-term utility if the follower accepts. In the very first step, before any BR is known, the current best manipulation in the longEU calculation cannot be determined and we set, for each leader's action $l$, $ u^L \left(X^L(\hat{l},\hat{f})-\hat{\bm{c}} \right) = u^L \left(X^L(l,\hat{BR}(l)\right)$.

\subsubsection{Probabilistic Feasible Region (PFR)}
To approximate $\mathbb{P} \left[ \text{accept} | l,f,\bm{c},W \right]$, we can use a probabilistic surrogate model for the follower's utility function. One can assign a uniform distribution to the feasible region and compute 
\begin{align} \label{eq:prob_accept_PFR}
    & \mathbb{P} \left[ \text{accept} | l,f,\bm{c} ,W \right] \\
    & := \int \mathds{1}\left[  u^F(X^F(l,f)+\bm{c}; W) \geq u^F(X^F(l,BR(l));W) \right] dW. \notag
\end{align}
A general framework of the longEU policy with PFR (see Algorithm \ref{alg:longEU_policy_PFR}) starts from initialising the feasible region model from the pairwise comparison data. For each round $t$, the policy finds the best manipulation $(l^t,f^t,c^t)$, i.e., the one with the highest longEU value, based on the feasible weight region and round information. The policy then observes whether the manipulation succeeds or not ($a^t==f^t$) and updates the result of manipulation to $\mathcal{D}^t$, the new feasible region $FR(\bm{w};\mathcal{D}^t)$ and the current best manipulation $(\hat{l},\hat{f},\hat{\bm{c}})$.

\begin{algorithm} 
\caption{longEU policy with PFR $^*$ }\label{alg:longEU_policy_PFR}
\begin{algorithmic}[1]
\Require $T>0$, $X^L, X^F$, $\bm{w}^L$, $\mathcal{C}$ 
    \State Set $t=0$, $\mathcal{D}^t=\emptyset$ and initiate $FR(\bm{w};\mathcal{D}^t)$ 
\For{$t=1$ to $T$}
    
    \State Pick $(l^t,f^t,c^t)= \underset{l,f,\bm{c}}{\mathrm{argmax}} \ \text{longEU}(l,f,\bm{c} | t,T, W)$ 
    \State Observe follower's action $a^t$ 
    \State Update $\mathcal{D}^t$, $FR(\bm{w};\mathcal{D}^t)$ and  $(\hat{l},\hat{f},\hat{\bm{c}})$
\EndFor

\end{algorithmic}
\scriptsize{$^*$longEU value can be replaced by EU value from \eqref{eq:EU}}
\end{algorithm}

\subsubsection{Random-weight Minimal Cost (RWMC)}
Integrating over the feasible weight space in Eq.~\ref{eq:prob_accept_PFR} may be computationally expensive. To reduce the computational cost, instead of integrating over the feasible weight space, one can pick a random weight $\hat{\bm{w}}_{r}$ from the feasible region $FR(\bm{w};\mathcal{D})$ and find the optimal manipulation cost $\hat{\bm{c}}_r(l,f)$ for each pair of actions $(l,f)$ from \eqref{eq:optimal_policy} with the substitution of the chosen weight $\hat{\bm{w}}_{r}$ into the follower's utility function $u^F$.  This will result in a manipulation that makes the manipulated action equally attractive to the BR action to the follower, given the weight $\hat{\bm{w}}_{r}$. We then simply assume that there is a 50\% probability that the follower accepts the offer
\begin{equation} \label{eq:prob_accept_RW}
\mathbb{P} \left[ \text{accept} | l,f,\hat{\bm{c}}_r(l,f), \hat{\bm{w}}_{r} \right]=0.5.
\end{equation}

Geometrically, a chosen manipulation defines a cutting plane in the feasible weight space that separates the weight vectors with which the follower will accept from those with which the follower would reject the manipulation. A natural choice for the cost seems to be the minimal cost $\hat{\bm{c}}_r(l,f)$, the smallest (in terms of leader's utility) cost needed to cut the feasible region given the randomly chosen weight. 

After the follower has rejected or accepted the offer, the feasible weight space is reduced by keeping only the corresponding side of the cutting plane. Ideally, if we could manage to pick a cutting plane that cuts the feasible weight space in half, it would be guaranteed to shrink logarithmically (divided by two in every round). In order to move the cutting plane more into the centre of the feasible weight space, we thus propose also a 
 variant of RWMC that calculates the centroid $\hat{\bm{w}}_{mid}$ from the feasible region as the chosen weight, and a corresponding approach is named Middle-Weight Minimal Cost (MWMC). In particular in 2-dimensional problems, the feasible region can be mapped to a 1-dimensional hyperplane due to $(w_1,w_2)=(w,1-w)$, and the interval of the first weight is $FR(w_1;\mathcal{D})=[w_{min},w_{max}]$. Hence, to effectively shrink the size of the feasible region, one can apply the idea of \textbf{Binary Search} and pick the middle point $W_{mid,1}=(w_{min}+w_{max})/2$. For $D>2$, the centroid is estimated as the average of a set of randomly generated samples produced using an MCMC sampling method \citep{Zabinsky2013}.
 
 Nevertheless, for RWMC and MWMC, we still need some criteria for choosing the manipulated action.
Algorithm \ref{alg:longEU_policy_RWMC} is modified from Algorithm \ref{alg:longEU_policy_PFR} and applies the highest longEU value as a criterion to select a manipulated action with a corresponding minimal cost (line 4-7). Additionally, the policy considers the tradeoff of exploring a new weight and exploiting the current best manipulation to not lose too much utility (line 8-13). Note that the MWMC approach has a chance to stop exploring a new weight if the longEU (or EU) value of the non-BR is smaller than that of the BR, or if the OMP is infeasible (not possible to manipulate successfully, especially for the (C2) constraint). Such situations do not help shrinking the feasible weight region. Thus, if the MWMC approach gets stuck with the same weight as in the previous round, we simply randomly choose a new weight in each round until we have chosen to manipulate again.
    
\begin{algorithm} 
\caption{longEU policy with RWMC (or MWMC) $^*$}\label{alg:longEU_policy_RWMC}    
\begin{algorithmic}[1]

\Require $T>0$, $X^L, X^F$, $\bm{w}^L$, $\mathcal{C}$
    \State Set $t=0$, $\mathcal{D}^t=\emptyset$ and initiate $FR(\bm{w};\mathcal{D}^t)$ 
\For{$t=1$ to $T$}
    \State Randomly generate weight $\hat{\bm{w}}_{r}$ from $FR(\bm{w};\mathcal{D}^t)$
    \Statex \Comment{if MWMC, pick $\hat{\bm{w}}_{r}=\hat{\bm{w}}_{mid}$ }
    \ForAll{ $(l,f) \in \mathcal{A}^L \times \mathcal{A}^F$ }
        \State Compute the minimal cost $\hat{\bm{c}}_{r}(l,f) \in \mathcal{C}$ from \eqref{eq:optimal_policy} 
        \State with the substitution of $l,f,\hat{\bm{w}}_{r}$
    \EndFor
     \State $(l_r,f_r)= \underset{l,f}{\mathrm{argmax}} \ \text{longEU}(l,f,\hat{\bm{c}}_{r}(l,f) | t,T, \hat{\bm{w}}_{r})$ 
     \State Compute $$\text{longEU}(\hat{l},\hat{f},\hat{\bm{c}} | t,T)= u^L \left(X^L(\hat{l},\hat{f})-\hat{\bm{c}} \right) \cdot (T-t+1) \label{eq:longEU_currbest}$$ 
     \If {$ \text{longEU}(l_r,f_r,\hat{\bm{c}}_{r}) | t,T, \hat{\bm{w}}_{r}) < \text{longEU}(\hat{l},\hat{f},\hat{\bm{c}} | t,T ) $ }
        \State Pick $(l^t,f^t,\bm{c}^t)=(\hat{l},\hat{f},\hat{\bm{c}})$
    \Else 
        \State Pick $(l^t,f^t,\bm{c}^t)=(l_r,f_r,\bm{c}_{r}(l_r,f_r))$
     \EndIf
    \State Observe follower's action $a^t$ 
    \State Update $\mathcal{D}^t$,$FR(\bm{w};\mathcal{D}^t)$ and  $(\hat{l},\hat{f},\hat{\bm{c}})$
\EndFor

\end{algorithmic}
\scriptsize{$^*$longEU value can be replaced by EU value from \eqref{eq:EU}. Then line \ref{eq:longEU_currbest} becomes $\text{EU}(\hat{l},\hat{f},\hat{\bm{c}})= u^L \left(X^L(\hat{l},\hat{f})-\hat{\bm{c}} \right) $}
\end{algorithm}

\subsection{Asymptotic Analysis of the longEU policy}
In every round, the longEU policy aims to improve its utility over the current best found manipulation regardless of the form of the parametric utility function or cost constraint. 
\begin{lemma} \label{thm:better_cost_longEU}
    If $BR(l)$ are known for all $l$, the maximiser $(l,f,\bm{c})$ of \eqref{eq:longEU} satisfies $u^L(X^L(l,f)-\bm{c})\geq u^L(X^L(\hat{l},\hat{f})-\bm{\hat{c}})$.
\end{lemma}
\begin{proof}
For convenience, we write the utility of the current best manipulation at time $t$, $(\hat{l},\hat{f},\hat{\bm{c}})$, as $u_{best}^t$ and denote $p:=\mathbb{P} \left[ \text{accept} | l,f,\bm{c} ,W \right]$.  Since $\mathbb{P} \left[ \text{accept} | \hat{l},\hat{f},\hat{\bm{c}} ,\bm{w} \right] =1 $, we can write  longEU also as 
\begin{align} \label{eq:aug_longEU}
    &\text{longEU}( l,f,\bm{c} \ | \ t,T,W) - \text{longEU}( \hat{l},\hat{f},\hat{\bm{c}} \ | \ t,T,W) \notag \\ 
    & = (T-t) p \left( u^L (X(l,f)-\bm{c}) - u_{best}^t  \right) \\
    &+ \left(u^L (X(l,f)-\bm{c}) p + u^L (X(l,BR(l))) (1 -p)\right) - u_{best}^t. \notag 
\end{align}
This term is negative for $u^L(X^L(l,f)-\bm{c})< u^t_{best}$, since $u^t_{best}\geq u^L(X(l,BR(l)))$ for all $l,t$.
For $(l,f,c)=(\hat{l},\hat{f},\hat{c})$ the above term is equal to zero, so zero can always be achieved and thus longEU will either propose the current best manipulation, or one with a larger utility, i.e., $u^L(X^L(l,f)-\bm{c})\geq u^t_{best}$. 
\end{proof}

When $T \rightarrow \infty$ and $p>0$, the calculation of \eqref{eq:aug_longEU} is dominated by the first term. The asymptotic form of the longEU policy, therefore, becomes 
\begin{align} \label{eq:asymp_longEU}
    & \text{longEU}^+( l,f,\bm{c} \ | W) := \\
    & \left({u^L \left(X^L(l,f)-\bm{c} \right)}- u^L \left(X^L(\hat{l},\hat{f})-\hat{\bm{c}} \right)\right) \cdot \mathbb{P} \left[ \text{accept} | l,f,\bm{c} ,W \right]. \notag
\end{align}
This longEU$^+$ form resembles the Expected Improvement (EI) algorithm \citep{jones1998efficient} in the sense that the new expected utility is compared to the current best solution. The following two theorems show that both algorithms converge to the optimal manipulation cost. 

 \begin{theorem} \label{thm:alg1_converge}
    Under constraint (C1), linear utility, and $D=2$, assume that  all BRs are known and $w^L_d,w^F_d>0$ for all $d \in [D]$. If $T \rightarrow \infty$, the manipulation cost from Algorithm \ref{alg:longEU_policy_PFR} will converge to the optimal manipulation cost of \eqref{eq:optimal_policy} for any beneficial games.
\end{theorem}

\begin{theorem} \label{thm:alg2_converge}
    Under constraint (C1) and linear utility, assume that all BRs are known and $w^L_d,w^F_d>0$ for all $d \in [D]$. If $T \rightarrow \infty$, the manipulation cost from Algorithm \ref{alg:longEU_policy_RWMC} will converge to the optimal manipulation cost of \eqref{eq:optimal_policy} for any beneficial game.
\end{theorem}
\noindent
Proofs can be found in the Appendix \ref{sec:appendix_proof}.

\section{Experiments}
In this section, we investigate the average performance of all policies over a large number of randomly generated games, and then have a closer look at their behaviour at some specifically chosen games. 

We compare the following policies:
\begin{itemize}
    \item no manipulation: Choose the best action without manipulation  given the known BR
    \item longEU + PFR : Algorithm \ref{alg:longEU_policy_PFR}
    \item longEU + RWMC : Algorithm \ref{alg:longEU_policy_RWMC} with weight randomisation
    \item longEU + MWMC : Algorithm \ref{alg:longEU_policy_RWMC} with centroid weight 
    \item EU + PFR : Algorithm \ref{alg:longEU_policy_PFR}  replacing longEU by EU calculation.
    \item EU + MWMC : Algorithm \ref{alg:longEU_policy_RWMC}  with centroid weight and replacing longEU by EU calculation.
    \item longEU$^+$+ MWMC : Algorithm \ref{alg:longEU_policy_RWMC} with centroid weight and replacing longEU by longEU$^+$ calculation.
\end{itemize}

Generally, the longEU policy starts with an exploration phase to learn the follower's weight. Once a new manipulation is deemed not worth exploring, the policy switches to exploiting the current best manipulation. The EU policy is more myopic, so it is unlikely to do as much exploration as the longEU policy. 

\subsection{Uniformly Random Game}
In this experiment, the payoff matrices as well as the utility weights are uniformly randomised. The payoff value for each objective is bounded between 0 and 1. We found that only approximately 30-40$\%$ of 2-objective (2D) and 3-objective (3D) random games are beneficial games. The average cumulative regret $\pm$ standard error over replications for different maximum numbers of rounds  $T$ are reported (note that longEU uses knowledge about $T$ and thus we here report only on the final cumulative regret given $T$). For $D>2$, MCMC sampling \cite{Zabinsky2013} is used to approximate the probability of accepting a cost.

The results show only very small differences between using constraints (C1) and (C2); thus, we only show the results with constraint (C2) here. We also run experiments on high-dimensional games ($D=10$). However, the trend of cumulative regret from all policies is generally similar to that of 2D cases. Due to space limit, we defer the results with constraint (C1) as well as results on high-dimensional games ($D=10$) can be found in the Appendix \ref{sec:appendix_results}.     

Figures~\ref{fig:2D-uniform-game:cumregret} and \ref{fig:3D-uniform-game:cumregret}  show that the best policy is longEU with the PFR model independent of the length of the game, even when $T$ is small. A closer look reveals that in some non-beneficial games, with the BR predicting model, the EU+PFR policy never corrects the wrong prediction of BR and keeps spending a manipulation cost on the actual BR. This scenario produces an unnecessary cost to the EU+PFR policy. Meanwhile, the longEU+PFR policy tries to improve the current best cost and eventually does not pay any cost to manipulate the actual BR. 

Considering the performance of Algorithm~\ref{alg:longEU_policy_RWMC}, as expected, choosing a middle weight (MWMC) provides a smaller cumulative regret than randomising (RWMC) the weight from the feasible region. Similar to the PFR model, a smaller cumulative regret of longEU over EU  is also observed from MWMC approaches even when the maximum time horizon $T$ is small. Besides, a larger number of available pairs of actions increases the cumulative regret of the EU policy more than of the longEU policy, especially in case of the MWMC approach, since a myopic perspective does not explore the follower's preference sufficiently, therefore having a higher chance of manipulating a suboptimal action.

\begin{figure}[h]

     \begin{subfigure}{0.5\textwidth}
         \centering
         \includegraphics[width=0.7\textwidth]{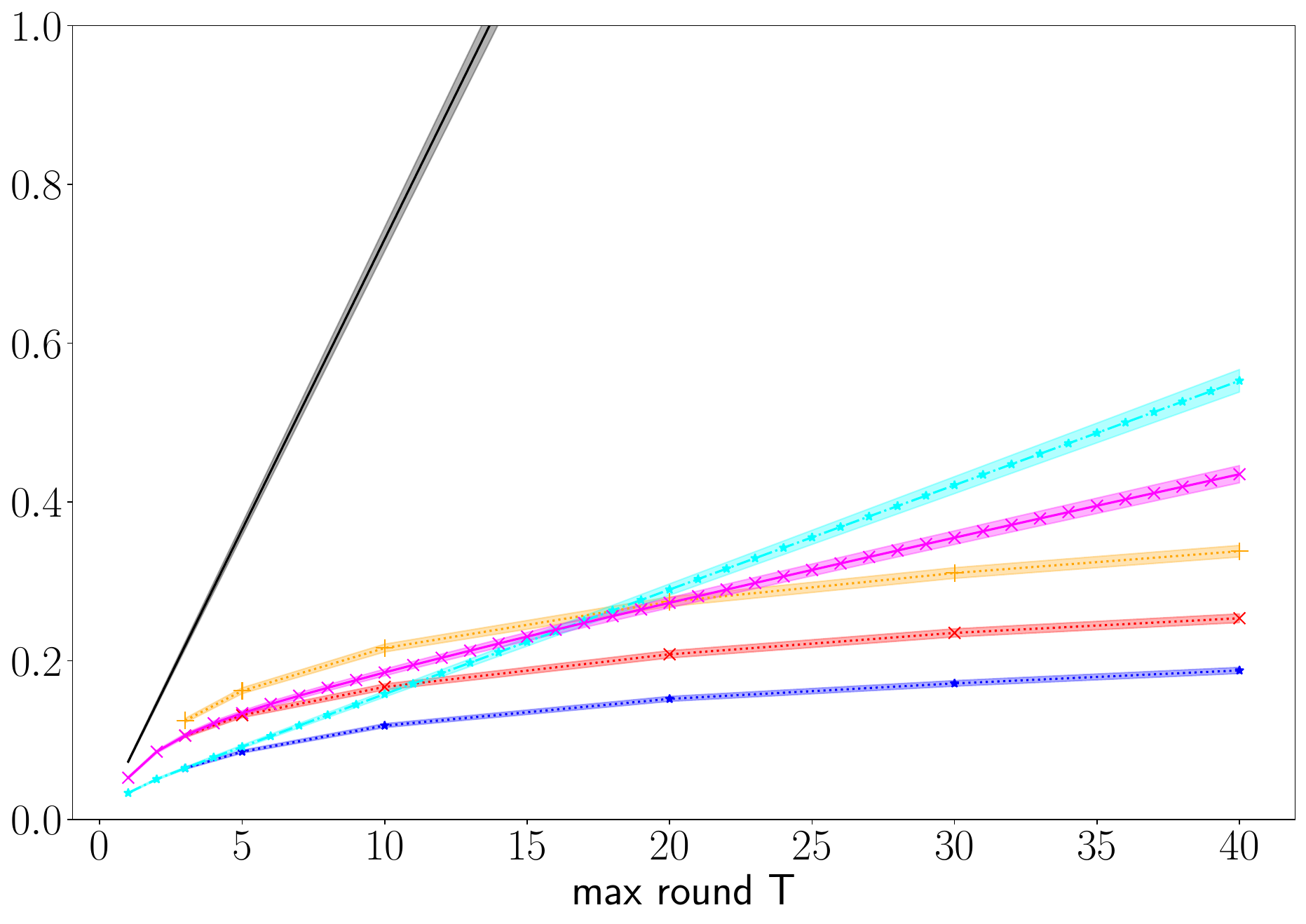}
         \caption{1x2 actions, constraint (C2) }
         \vspace{1em}
         \label{fig:2D-uniform-game:cumregret-1x2action-C2}
     \end{subfigure} 
  
     \begin{subfigure}{0.5\textwidth}
         \centering
         \includegraphics[width=0.7\textwidth]{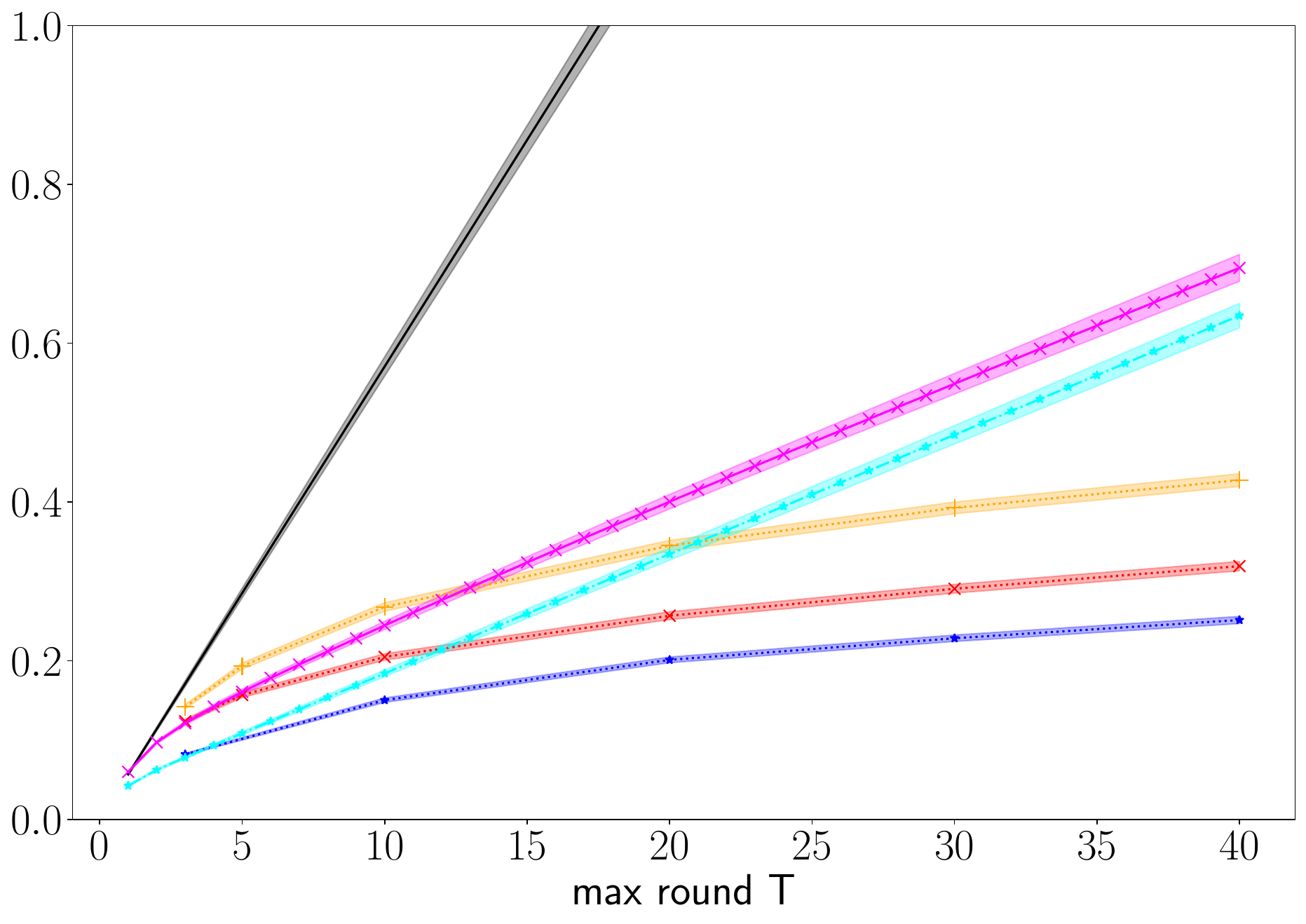}
         \caption{2x2 actions, constraint (C2) }
         \label{fig:2D-uniform-game:cumregret-2x2action-C2}
     \end{subfigure}
     \hfill
    \begin{subfigure}{0.5\textwidth}
         \centering
         \vspace{1em}
         \includegraphics[width=0.9\textwidth]{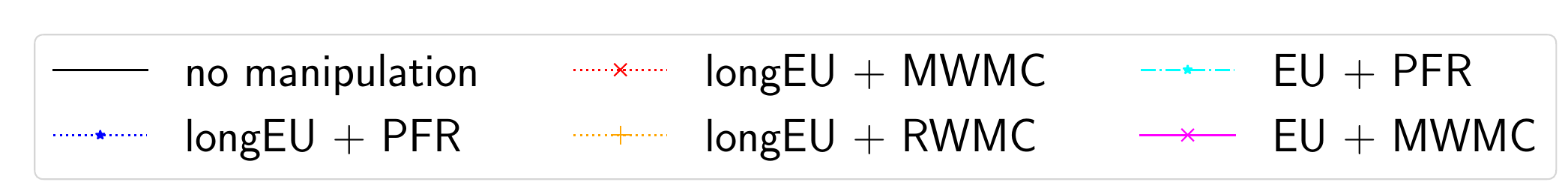}
     \end{subfigure}
    \caption{Final cumulative regret for 2D uniformly random games from 10000 replications}
    \label{fig:2D-uniform-game:cumregret}
    \vspace{1em}
\end{figure}

\begin{figure}[h]
    
     \begin{subfigure}{0.5\textwidth}
         \centering
         \includegraphics[width=0.7\textwidth]{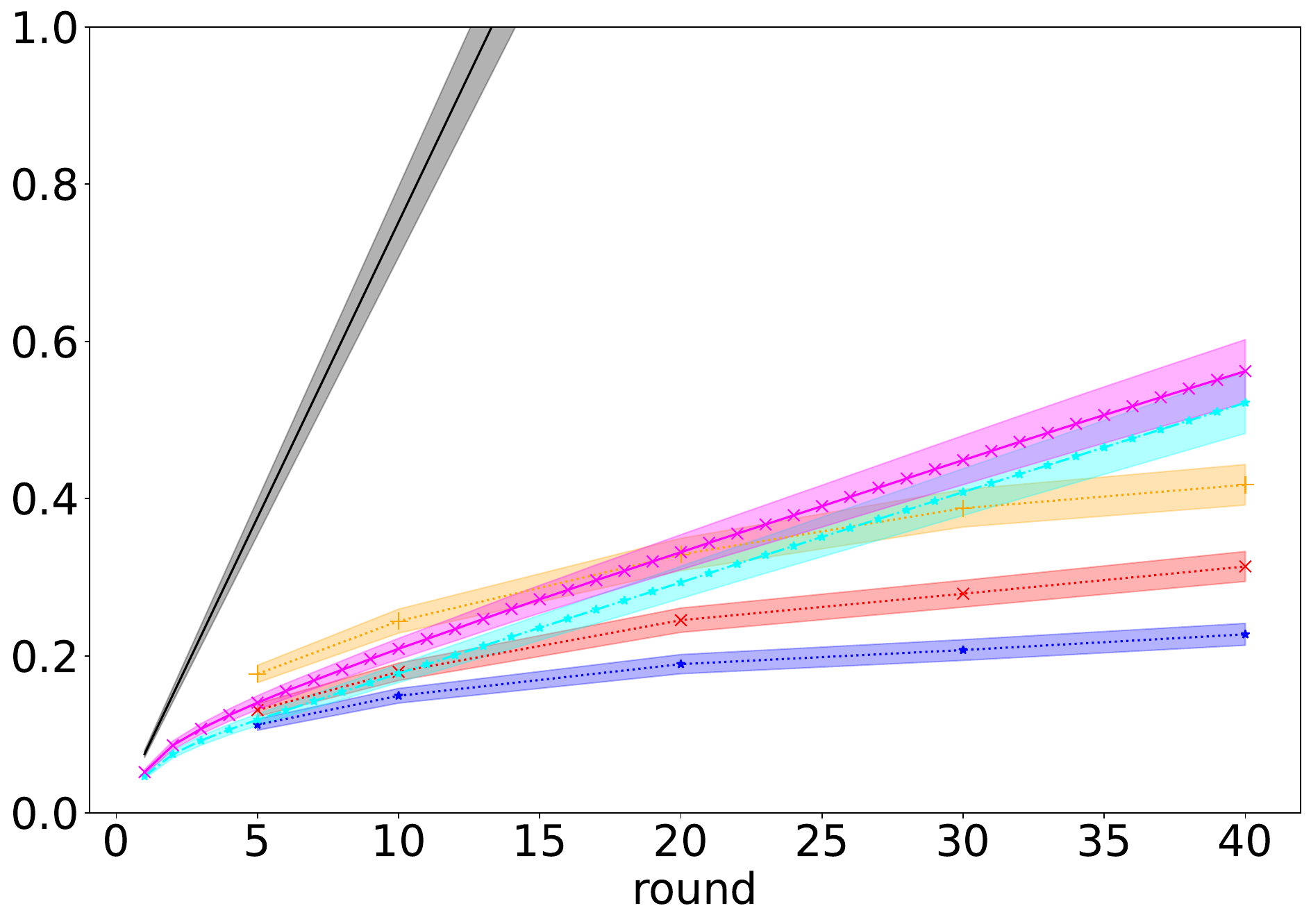}
         \caption{1x2 actions, constraint (C2) }
         \vspace{1em}
         \label{fig:3D-uniform-game:cumregret-1x2action-C2}
     \end{subfigure} 
  
     \begin{subfigure}{0.5\textwidth}
         \centering
         \includegraphics[width=0.7\textwidth]{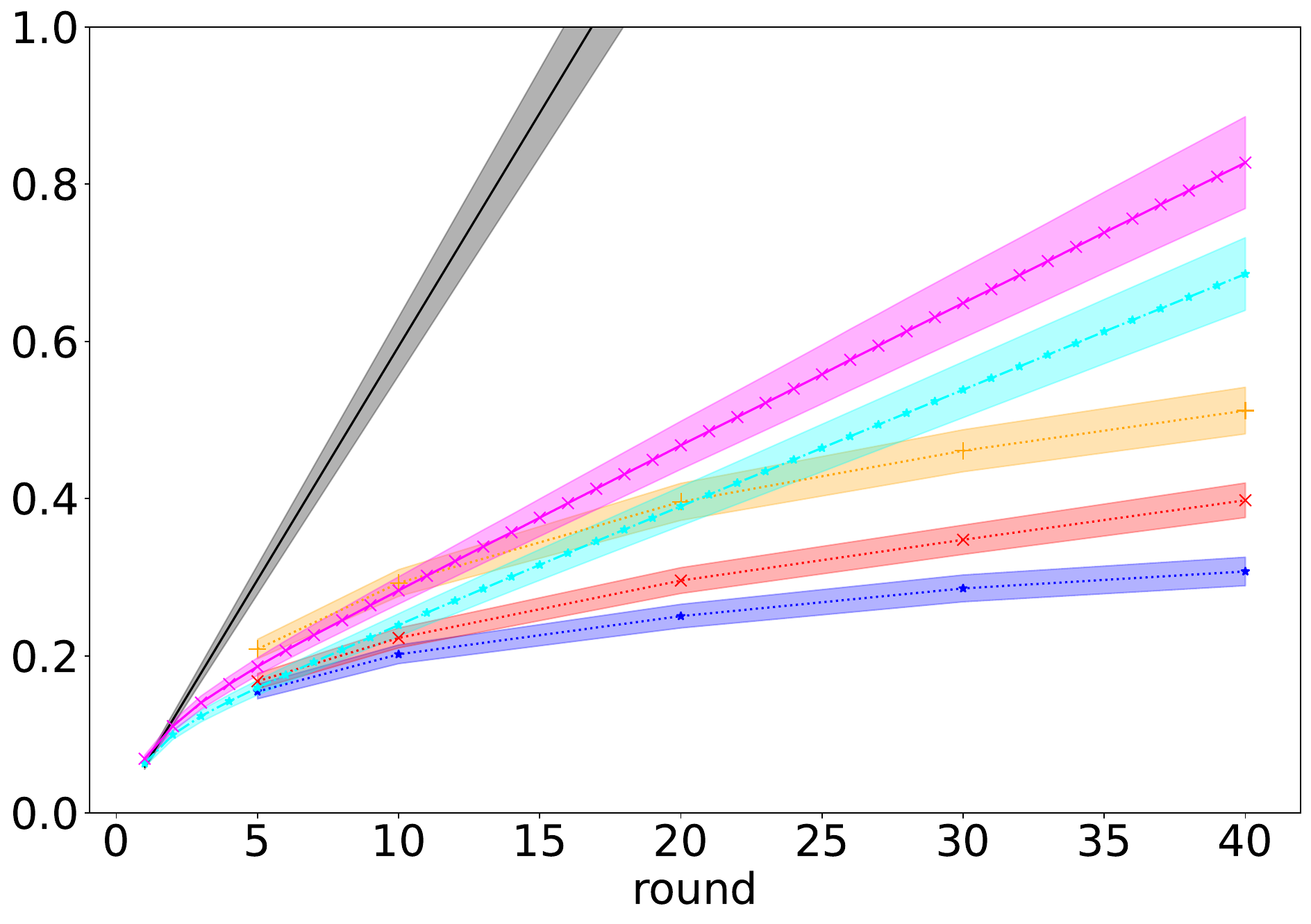}
         \caption{2x2 actions, constraint (C2) }
         \label{fig:3D-uniform-game:cumregret-2x2action-C2}
     \end{subfigure}
    \vspace{1em}
    \caption{Final cumulative regret for 3D uniformly random games from 1000 replications}
    \label{fig:3D-uniform-game:cumregret}
    \vspace{1.5em}
\end{figure}

 In the next two sections, we want to observe how each algorithm behaves in some specific situations.
 
\subsection{2D Game: Low-risk Low-return or High-risk High-Return} 

We now turn to investigate how the non-myopic perspective of longEU influences its manipulation choices.
In partiucular, our hypothesis is that when $T$ is large, longEU policies are willing to go for the high-risk high-return utility instead of a low-risk low-return utility. 
To do so, we design the following game to test this:

\begin{center}
\text{Leader's Payoff Matrix} $ X^L=\begin{bmatrix}
[1,1] & [1.6,1.2] \\
[0.1,0.1] & [2.4,1] 
\end{bmatrix} $

\text{Follower's Payoff Matrix} $ X^F=\begin{bmatrix}
[1,1] & [0.8,0.2] \\
[1,1] & [0.8,0.2] 
\end{bmatrix}  $ \\
where $\bm{w}^L = [0.4,0.6]$ and $\bm{w}^F = [0.8,0.2]$ .
\end{center}
For the leader, under positive cost constraint (C1), we have
\begin{itemize}
    \item Leader's action $l=1$ is the best action without manipulation, which gives a utility of 1 as the follower will choose BR(1)=1. It is also the low-risk low-return manipulated action by spending [0.4,0] to incentivise the follower to choose $f=2$ and which results in a utility of 1.2 for the leader. 
    \item Leader's action $l=2$ allows for a high-risk, high-return manipulated action by spending the optimal cost of [0.4,0] to get the utility of 1.4, but if the manipulation fails, they get 0.1.
\end{itemize}

\begin{figure}[h]
     \begin{subfigure}{0.5\textwidth}
         \centering
         \includegraphics[width=0.7\textwidth]{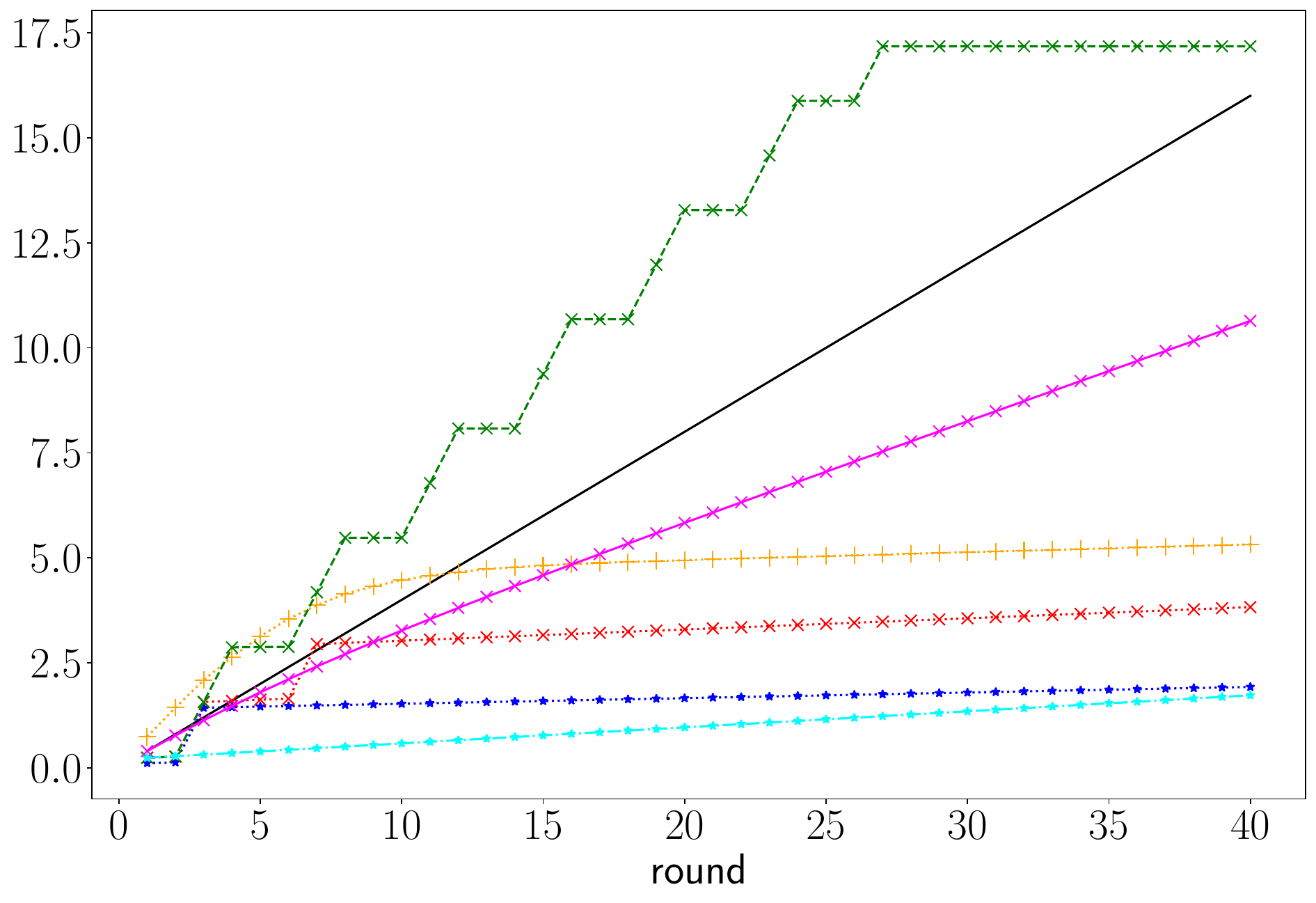}
         \caption{Cumulative Regret }
         \label{fig:2D-high-risk:cumregret}
         \vspace{1em}
     \end{subfigure}
    \hfill
     \begin{subfigure}{0.5\textwidth}
         \centering
         \includegraphics[width=0.7\textwidth]{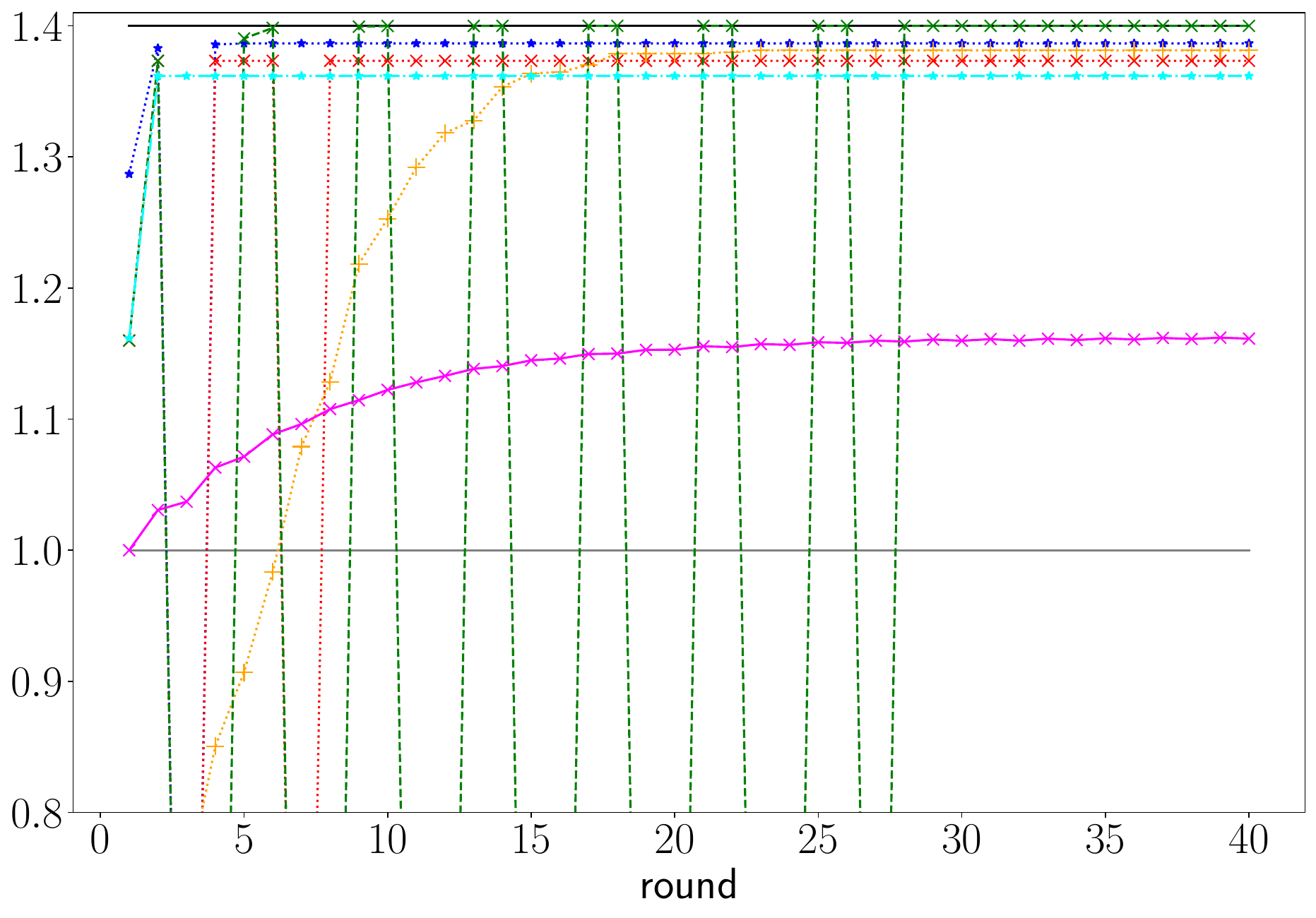}
         \caption{Leader's Utility}
         \label{fig:2D-high-risk:utility}
     \end{subfigure}
     \hfill 
     \begin{subfigure}{0.5\textwidth}
         \centering
         \vspace{1em}
         \includegraphics[width=0.9\textwidth]{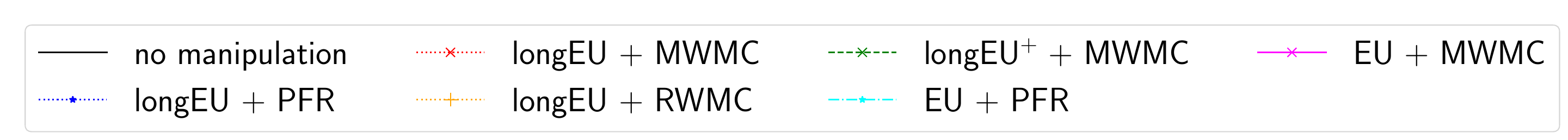}
     \end{subfigure}
    \hfill 
    \caption{Performance measures on 2D high-risk-high-return Game of different policies given $T=40$ }
    \label{fig:2D-high-risk}
    \vspace{2em}
\end{figure} 
\noindent With these specific games, we want to obtain a better insight into the behaviour of the different policies over the run, so we report on (a) the cumulative regret over the run, and (b) the leader's utility for each iteration, for a single fixed horizon $T$ (this makes a difference only for longEU which needs to know $T$). From Figure \ref{fig:2D-high-risk}, the policies' behaviour can be summarised as follows
\begin{itemize}
    \item Overall, the cumulative regret trend is quite consistent with the random game. However, EU+PFR  is the best policy with the smallest cumulative regret in this case. 
    \item EU + PFR  starts from the low-risk-low-return manipulated action. Once this manipulation has been successful, it then uses the same cost to manipulate the high-risk-high-return action (as the follower's payoffs are identical, it is clear that the follower will accept the manipulation also in the high risk case).
     \item longEU +PFR  starts from the high-risk-high-return manipulated action since it is more attractive from the long-term perspective. The policy lowers the cost to manipulate such an action, but it fails once at round 3. Shortly after that, the policy stops exploring and offers a higher than necessary manipulation cost. Note that this is not a contradiction to the consistency of longEU+PFR, but due to the fact that the policy 'knows' that it has a limited budget, so it is less explorative.
     \item EU + MWMC  tends to manipulate the low-risk-low-return action since it believes to fail at 50$\%$, and the loss for the high-risk-high-return action is much higher than a gain. 
    \item longEU + RWMC policy performs worse than longEU + MWMC  since it is slower to improve the feasible region in the exploration phase.
    \item longEU$^+$+ MWMC explores a new weight until the true weight is found, so it alternates between underestimation and overestimation of the actual weight. This leads to the alternation between success and failure to manipulate.
    
\end{itemize}

\subsection{2D Game: Play Safe or Not} 
We now investigate what happens if the algorithms get stuck in a non-informative manipulation cost called the play-safe cost. 
Such a cost ensures a 100$\%$ acceptance rate from the follower without gaining information about the feasible region (i.e., the algorithms will perform suboptimally).
Our goal is to understand which algorithms would perform the best in this worst-case scenario.
To do so, we design the following game:
\begin{center}
\text{Leader's Payoff Matrix} $ X^L=\begin{bmatrix}
[1,1] & [1.3,1.3] \\
[0.1,0.1] & [2,2] 
\end{bmatrix} $
\text{Follower's Payoff Matrix} $ X^F=\begin{bmatrix}
[1,1] & [0.9,0.9] \\
[1,1] & [0.5,0.5] 
\end{bmatrix}$ \\
where $\bm{w}^L = [0.6,0.4]$ and $\bm{w}^F = [0.8,0.2]$ .
\end{center}
For the leader, under positive cost constraint (C1), we have
\begin{itemize}
    \item Leader's action $l=1$ is the best action without manipulation, which gives a utility of 1. This action is also the low-risk low-return manipulated action by spending [0.1,0] to get the utility of 1.24. However, one can play safe by spending [0.1,0,1] to surely get the utility of 1.2. 
    \item Leader's action $l=2$ is the high-risk, high-return manipulated action by spending the optimal cost of [0.625,0] to get a utility of 1.625, but if the manipulation fails, they get the utility 0.1. However, one can play safe by spending [0.5,0,5] to surely get 1.5. 
\end{itemize}
\begin{figure}[h]
     \begin{subfigure}{0.5\textwidth}
         \centering
         \includegraphics[width=0.7\textwidth]{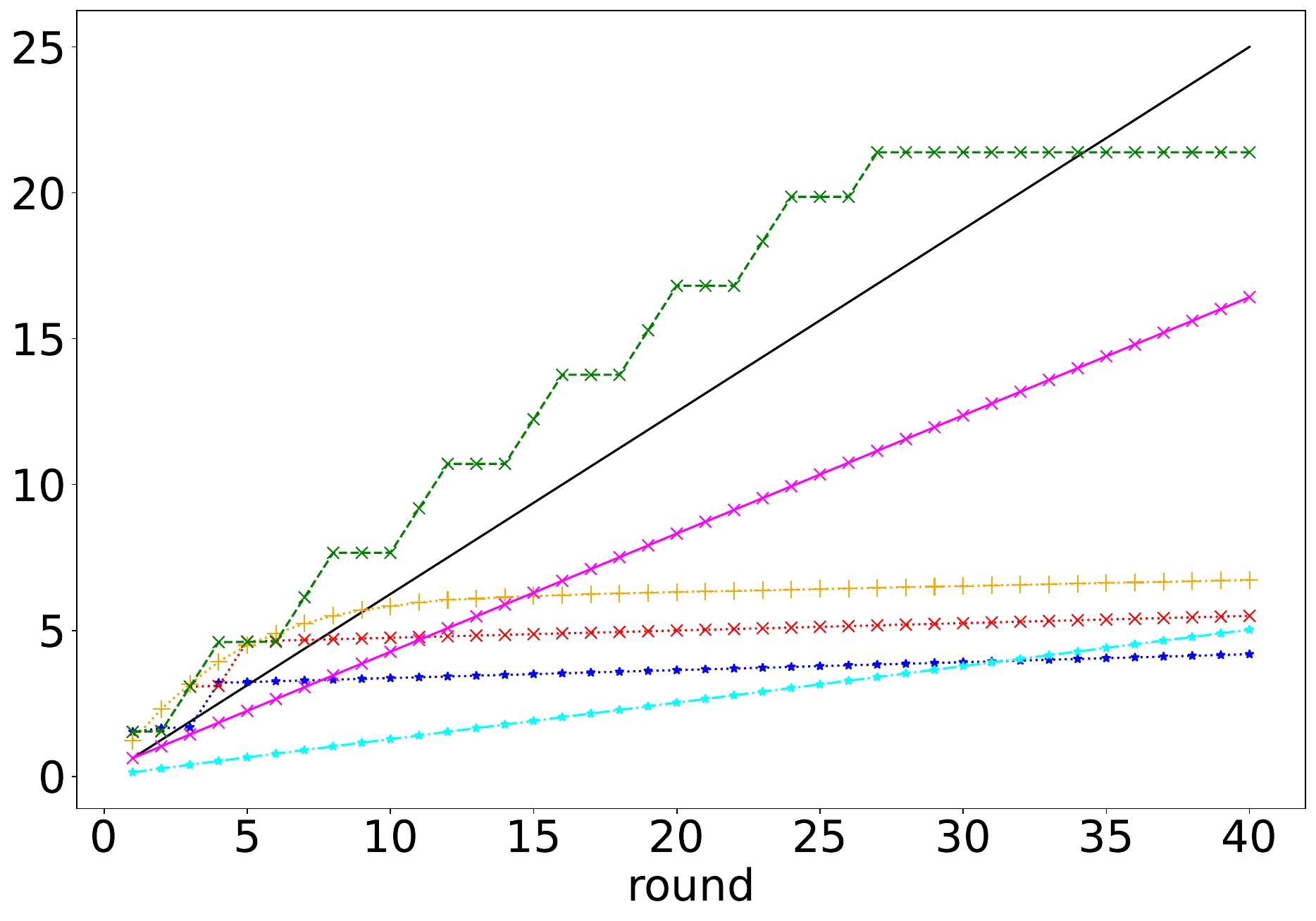}
         \caption{Cumulative Regret }
         \vspace{1em}
         \label{fig:2D-play-safe:cumregret}
     \end{subfigure}
    \hfill 
     \begin{subfigure}{0.5\textwidth}
         \centering
         \includegraphics[width=0.7\textwidth]{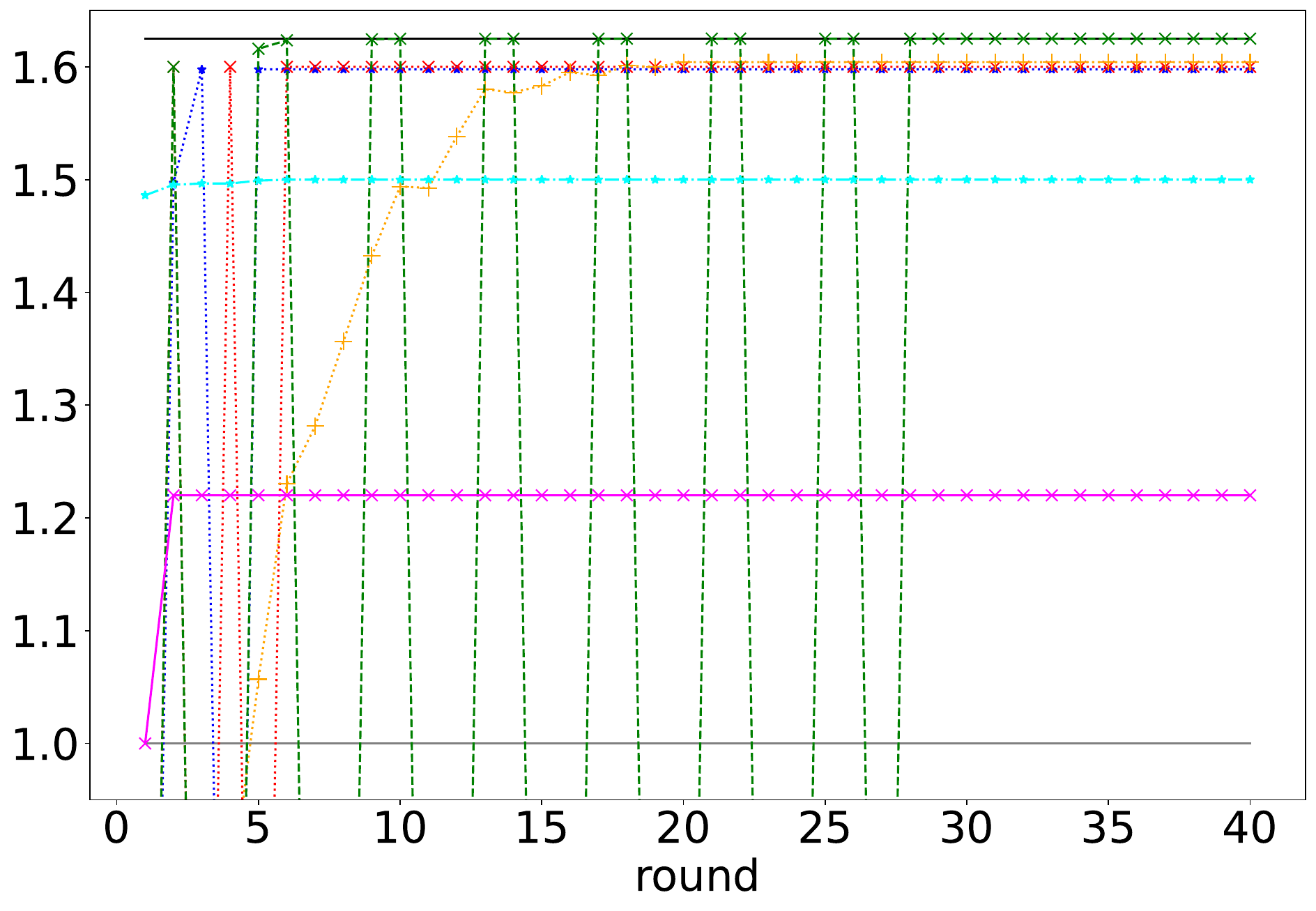}
         \caption{Leader's Utility}
         \label{fig:2D-play-safe:utility}
     \end{subfigure}
    \hfill 
    \caption{Performance measures on 2D Play-safe-or-not Game of different policies given $T=40$}
        \label{fig:2D-play-safe}
    \vspace{2em}
\end{figure}

\noindent From Figure \ref{fig:2D-play-safe}, longEU + PFR  is the best policy and only  EU + PFR  decides to play safe to receive a utility of 1.5.

\section{Conclusion and Future Directions}
In this paper, we investigated payoff manipulation in multi-objective repeated Stackelberg games. Incentivising a rational follower by sacrificing a fraction of a vector payoff to steer the follower's action towards the leader's preference can lead to a higher utility for both, the leader and the follower. To identify the best manipulation, the leader needs to learn about the follower's utility function.  

Assuming the utility function is linear, the follower's preference can be constrained by estimating the feasible weight region given the history of manipulation outcomes. We proposed the longEU policy, a non-myopic manipulation policy, to effectively incentivise the follower while minimising the cumulative regret of the leader. In empirical tests, the proposed longEU policy generally achieved a lower cumulative regret than the EU policy, a myopic manipulation policy, if the game is repeated over a sufficient number of iterations. Assuming the follower reacts deterministically with the BR, we simply model the follower's possible weight by a feasible region constrained by linear inequalities given the follower's preferred action on the leader's manipulation. We introduce two approaches to calculate the longEU value; 1) PFR: assigning a uniform distribution over the feasible region to compute the probability of acceptance 2) RWMC (and MWMC): computing the minimal manipulation given a weight first and assigning the probability of accepting the minimal cost as half. The PFR approach shows the best performance over random games and pre-defined games. One limitation of the PFR approach is that it requires substantial computational time in high-dimensional problems with a large action set. To deal with the large action set, one can ignore actions which are likely to provide a lower utility than the current best utility. The second approach can accelerate the optimisation process, but results in a higher cumulative regret. Nevertheless, empirically, both approaches with longEU policy seem to have a sublinear growth. We also prove their convergence to the optimal manipulation in an infinite game. 

There are many avenues for future research. The payoff manipulation allows for some sort of collaboration of rational agents because the offered incentive from the leader provides at least a higher utility compared to the best response given the leader's action, and there is no reason to reject a good offer. However, if the follower is also a learner, the advantage of payoff manipulation for the leader may be reduced. In order to maximise their utility, the follower could deliberately reject some offers to make the leader think that the offer was not good enough and improve on the offer, or the follower may implement a mixed strategy to avoid the leader learning their utility function directly. Finally, we can consider penalties for the follower instead of incentives.


\begin{ack}
The first author gratefully acknowledges support by the Engineering and Physical Sciences Research Council through the Mathematics of Systems II Centre for Doctoral Training at the University of Warwick (reference EP/S022244/1).
\end{ack}


\bibliography{mybibfile}

\onecolumn
\appendix
\section*{Appendix}
\section{Mathematical proof} \label{sec:appendix_proof}
\subsection{Proof of Lemma \ref{thm:optimal_cost_C1}}
\begin{proof}
    Since the number of pairs of actions $(l,f)$ is finite, we can consider the optimal cost for each pair of actions and compare the utilities of such manipulations directly. Since the utility functions are linear, then for all pairs of action $(l,f)$, the (OMP) problem is a linear program. Thus there is a solution in one of the corner points of the feasible region. For any $(l,BR(l)))$, the optimal cost is $\bm{c}^*=\bm{0}$. For any pair of leader's action and follower's non-BR, the hyperplane of constraint (IC) creates a cutting point on all axes at $\bm{c}^*=(0,...,0,c_d,0,...,0)$ where $c_d\geq 0$ for all $d \in [D]$. Hence the solution is one of these cutting points. 
\end{proof}

 \subsection{Proof of Lemma \ref{thm:optimal_cost_C2}}
 \begin{proof}
    The proof follows the same idea as in Lemma \ref{thm:optimal_cost_C1}. It suffices to find the boundary points of feasible regions. For any $(k,BR(l))$, the optimal cost is $\bm{c}^*=\bm{0}$. For any pair of leader's action and follower's non-BR, the hyperplane of constraint (IC) either creates a cutting point on all edges of the hyperrectangle at $\bm{c}^*=(\bm{0},c_d,\bm{c}_{max})$ where $c_d\geq 0$  and $\bm{c}_{max}=(X^L_{d_1},X^L_{d_2},...,X^L_{d_k})$ for some $d,d_1,d_2,...,d_k \in [D]$, or has no intersection with the hyperrectangle. Hence the solution is either one of these cutting points or no manipulation. 
\end{proof}

\subsection{Proof of Theorem \ref{thm:alg1_converge}}
\begin{proof}
    Since all best response actions $BR(l)$ are known for all $l \in \mathcal{A}^L$, then at round $t=0$, $u^t_{best}=\max_{l \in \mathcal{A}^L} u^L(l,BR(l))$.  Since $T \rightarrow \infty$, the longEU calculation of \eqref{eq:longEU} is replaced by \eqref{eq:asymp_longEU}, longEU$^+$. 
    
    By Lemma \ref{thm:optimal_cost_C1}, the optimal cost is a single-objective cost $c^*$. Let $c_d^{min,t}(l,f)$ be the lower bound of a single-objective cost corresponding to the boundary of the $d^{th}$-objective estimated weight. The index $(l,f)$ refers to the value corresponding to the pair of actions $(l,f)$. Similarly, define $c_d^{max,t}(l,f)$ as the upper bound of $d^{th}$-objective cost to spend corresponding to the current best utility $u^t_{best}$.
    
    Since an interested game is a beneficial game, there exists $d^*,l^*,f^*$ such that $c_{d^*}^{min,t}(l^*,f^*)<c^*< c_{d^*}^{max,t}(l^*,f^*)$. In addition, we can restrict the search cost space $\mathcal{C}$ of this policy to a single-objective cost space. By Lemma \ref{thm:better_cost_longEU}, $c_{d^*}^{max,t}(l^*,f^*)$ is a decreasing sequence towards $c^*$. Meanwhile, $c_{d^*}^{min,t}(l^*,f^*)$ is an increasing sequence towards $c^*$ since the feasible region either is improved or remains the same. We only need to ensure that $c_{d^*}^{max,t}(l^*,f^*)-c_{d^*}^{min,t}(l^*,f^*) \rightarrow 0$\\

    \noindent \textbf{Smallest Case}:  \\
    Consider the simplest case where $|\mathcal{A}^L|=1$ and $|\mathcal{A}^F|=2$. Let $X^F:= [[x_1,x_2] \ [y_1,y_2]]$ be the follower's payoff matrix. WLOG, suppose the best-response is $BR(1)=1$. Since $D=2$, the feasible region is reduced to an interval $[w_d^{min,t},w_d^{max,t}]$. Define $w_d^{min,t},w_d^{max,t}$ as the minimum and maximum of the $d^{th}$-objective weight, respectively, where $w_2^{min,t}=1-w_1^{max,t}$ and  $w_2^{max,t}=1-w_1^{min,t}$. Note that we remove indices $(l,f)$ and simply write $c_1^{min,t},c_2^{min,t}$,$c_1^{max,t},c_2^{max,t}$ for convenience. Then $c_1^{min,t},c_2^{min,t}$ are updated based on $w_1^{min,t} (\text{or} \ w_1^{max,t}),w_2^{max,t} (\text{or} \ w_2^{max,t})$, respectively.
    
    \textbf{Fact:} There exists the minimal cost $\bm{c}^{o}:=(c^o_1,c^o_2)=X^F(1,1)-X^F(1,2)$, or the so-called \textit{play-safe cost}, to manipulate the follower into choosing action $f=2$ for all possible weights $\bm{w} \in FR(\bm{w};\emptyset)$. Therefore, we can separate the proof into three cases: (I) $c^o_1 > 0, c^o_2 > 0$, (II) $c^o_1 > 0, c^o_2 \leq 0$ and (III) $c^o_1 \leq 0, c^o_2 > 0$. It is impossible to have $c^o_1, c^o_2 < 0$ because the follower would not change its BR $l=1$. 

    \begin{figure} 
        \centering
        \includegraphics[width=\linewidth]{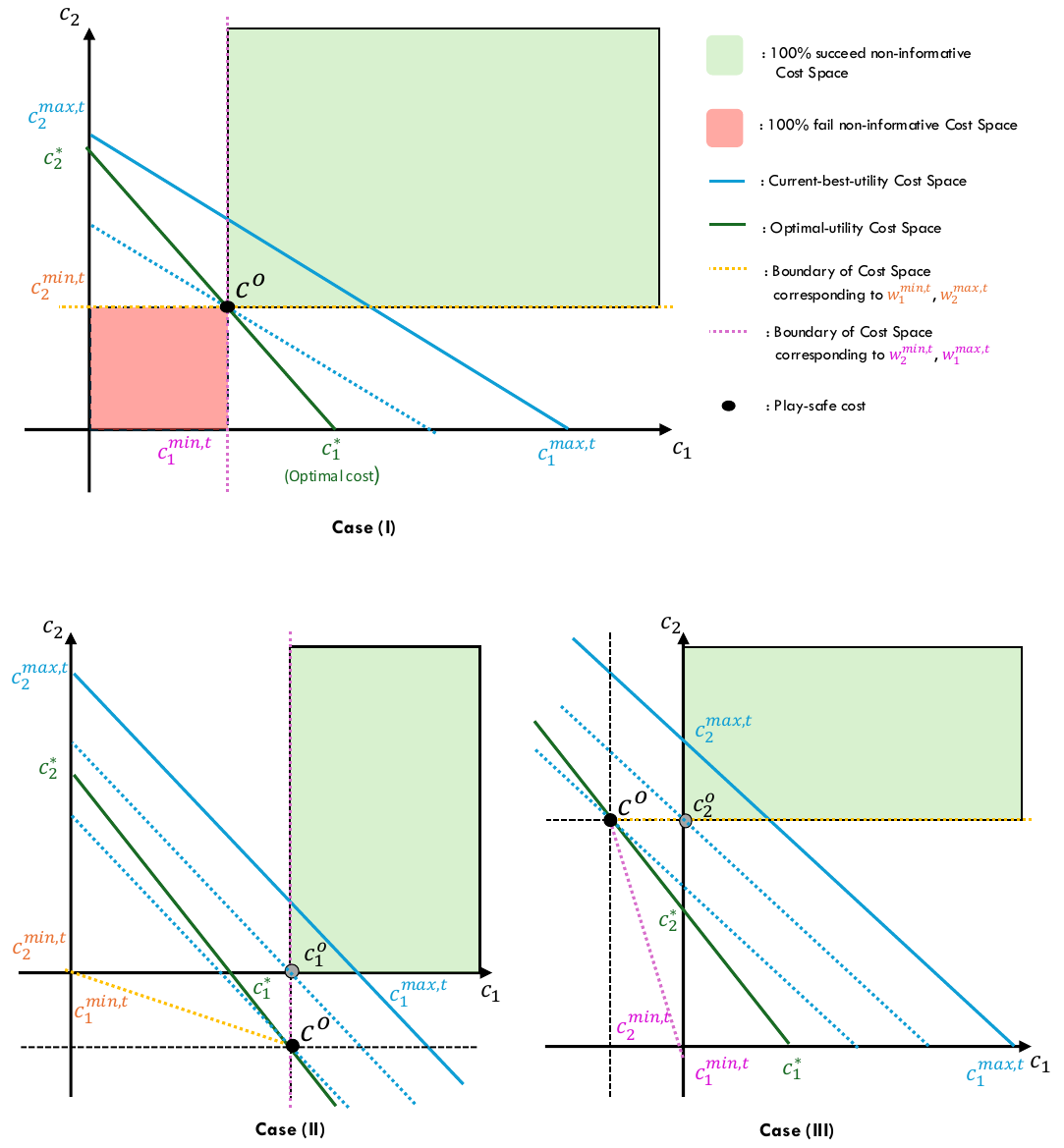}
        \caption{Play-safe Cost Cases}
        \label{fig:proof_2D_longEU_PFR}
    \end{figure}
    
    \textbf{Case (I)} For $c^o_1 > 0, c^o_2 > 0$:  
    At round $t\geq 1$, we will show that the optimised cost of longEU$^+$ is either
    \begin{itemize}
        \item $(\tilde{c}_1^t,0)$ such that $\tilde{c}_1^t= -\alpha+\sqrt{(\alpha+c_1^{min,t-1})(\alpha+c_1^{max,t-1})}$ or
        \item $(0,\tilde{c}_2^t)$ such that $\tilde{c}_2^t= \alpha+\sqrt{(-\alpha+c_2^{min,t-1})(-\alpha+c_2^{max,t-1})}$ where $\alpha=\Delta y- \Delta x$, $\Delta y=y_1-y_2$ and $\Delta x=x_1-x_2$.
    \end{itemize} 
    Denote $W^t$ as a uniformly-distributed random variable in $FR(\bm{w}| \mathcal{D}^t)$ and $p^t(\bm{c}^t)$  as the probability of accepting cost $\bm{c}^t:=(c_1^t,c_2^t)$. Firstly, we consider the longEU$^+$ function restricted to the first-objective cost. \\
  
  For $t\geq1$ and there is no successful manipulation,
  \begin{align*}
      \text{longEU}^+(1,2,(\tilde{c}_1,0) | W^t) & = \left({u^L \left(X^L(1,2)-(\tilde{c}_1,0) \right)}- u^L \left(X^L(1,2)-(c_1^{max,t-1},0) \right)\right) p^t((\tilde{c}_1,0))\\
      &=w_1^L(c_1^{max,t-1}-\tilde{c}_1)\left(\dfrac{w_1^{max,t}-\left(\dfrac{x_2 -y_2 }{\Delta y - \Delta x + \tilde{c}_1 }\right)}{w_1^{max,t}-w_1^{min,t}} \right) \\
      &=w_1^L(c_1^{max,t-1}-\tilde{c}_1)\left(\dfrac{\left(\dfrac{x_2 -y_2 }{\Delta y - \Delta x + c_1^{min,t-1} }\right)-\left(\dfrac{x_2 -y_2 }{\Delta y - \Delta x + \tilde{c}_1 }\right)}{\left(\dfrac{x_2 -y_2 }{\Delta y - \Delta x + c_1^{min,t-1}}\right)} \right) \\
      &=w_1^L(c_1^{max,t-1}-\tilde{c}_1)\left( \dfrac{\tilde{c}_1 -c_1^{min,t-1}}{\alpha + \tilde{c}_1 } \right)
  \end{align*}

  Setting its derivative to zero, its stationary point is $\tilde{c}_1= -\alpha + \sqrt{(\alpha+c_1^{min,t-1})(\alpha+c_1^{max,t-1})}$. 
  
  For $t\geq1$ and there is at least one successful manipulation, define $h_d(\bm{w}):=\dfrac{\bm{w}\cdot \bm{c}^o}{w_d}$ as a function returning the minimal cost from only objective $d$  which is required to manipulate the follower with respect to the follower's weight $\bm{w}$. Then $\text{longEU}^+(1,2,(\tilde{c}_1,0) | W^t)=w_1^L(c_1^{max,t-1}-\tilde{c}_1)\left( \dfrac{\tilde{c}_1 -c_1^{min,t-1}}{\alpha + \tilde{c}_1 } \right)\dfrac{\alpha + h_1((w_1^{max,t-1},w_2^{min,t-1}))}{h_1((w_1^{max,t-1},w_2^{min,t-1}))-c_1^{min,t-1}}$, and we obtain the same stationary point. 
  
  Similarly, if the longEU value from the second objective is higher, the leader will pay the cost of $(0,\tilde{c}_2)$.  The longEU$^+$ function is either $\text{longEU}^+(1,2,(0,\tilde{c}_2)=w_2^L(c_2^{max,t-1}-\tilde{c}_2)\left( \dfrac{\tilde{c}_2 -c_2^{min,t-1}}{-\alpha+ \tilde{c}_2 } \right)$ or  $\text{longEU}^+(1,2,(0,\tilde{c}_2)=w_2^L(c_2^{max,t-1}-\tilde{c}_2)\left( \dfrac{\tilde{c}_2 -c_2^{min,t-1}}{-\alpha + \tilde{c}_2 } \right)\dfrac{-\alpha + h_2((w_1^{min,t-1},w_2^{max,t-1}))}{h_2((w_1^{min,t-1},w_2^{max,t-1}))-c_2^{min,t-1}}$. Hence its stationary point is $\tilde{c}_2= \alpha +  \sqrt{(-\alpha+c_2^{min,t-1})(-\alpha+c_2^{max,t-1})}$.
    
    Finally, we will show that the manipulated cost will converge to the optimal cost. By AM-GM inequality, $\tilde{c}_1= -\alpha + \sqrt{(\alpha+c_1^{min,t-1})(\alpha+c_1^{max,t-1})} \leq (c_1^{min,t-1}+c_1^{max,t-1})/2$. The leader offers a  cost smaller than the average of cost interval $[c_1^{min,t-1},c_1^{max,t-1}]$. In other words, the relative gap from $\tilde{c}_1$ to $c_1^{min,t-1}$ is less than 0.5. Define the relative gap function from $\tilde{c}_1$ to $c_1^{min}$ as $g(c_1^{min},c_1^{max})=\dfrac{-\alpha + \sqrt{(\alpha+c_1^{min})(\alpha+c_1^{max})}-c_1^{min}}{c_1^{max}-c_1^{min}}$. Fixing $c_1^{min}$,  $g(c_1^{min},c_1^{max})$ is strictly increasing with respect to $c_1^{max}$. Oppositely, fixing $c_1^{max}$,  $g(c_1^{min},c_1^{max})$ is strictly decreasing with respect to $c_1^{min}$. Since $\{c_1^{max,t}\}_{t=0}$ is a decreasing sequence and $\{c_1^{min,t}\}_{t=0}$ is an increasing sequence, $\{g(c_1^{min,t},c_1^{max,t})\}_{t=0}$ is a positive increasing sequence bounded above by 0.5. A similar argument holds  for $\tilde{c}_2= \alpha +  \sqrt{(-\alpha+c_2^{min,t-1})(-\alpha+c_2^{max,t-1})}$.  When $t \rightarrow \infty$, the manipulation cost will make $c_1^{max,t}-c_1^{min,t} \rightarrow 0$ and  $c_2^{max,t}-c_2^{min,t} \rightarrow 0$. In other words, the optimal cost will be identified.

    \textbf{Case (II)} For $c^o_1 > 0, c^o_2 \leq 0$: The cost $(c^o_1,0)$ is the minimal cost restricted by the constraint (C1) to manipulate the follower with all possible weights. \text{The proof will follow the same idea as in Case (I).} We consider the longEU$^+$ function restricted to the first-objective cost again.
    
    For $t\geq1$ and there is no successful manipulation,
  \begin{align*}
      \text{longEU}^+(1,2,(\tilde{c}_1,0) | W^t) & = \left({u^L \left(X^L(1,2)-(\tilde{c}_1,0) \right)}- u^L \left(X^L(1,2)-(c_1^{max,t-1},0) \right)\right) p^t((\tilde{c}_1,0))\\
      &=w_1^L(c_1^{max,t-1}-\tilde{c}_1)\left(\dfrac{\left(\dfrac{x_2 -y_2 }{\Delta y - \Delta x + \tilde{c}_1 }\right)-w_1^{min,t-1}}{w_1^{max,t-1}-w_1^{min,t-1}} \right) \\
      &=w_1^L(c_1^{max,t-1}-\tilde{c}_1)\left(\dfrac{\left(\dfrac{x_2 -y_2 }{\Delta y - \Delta x + \tilde{c}_1 }\right)-\left(\dfrac{x_2 -y_2 }{\Delta y - \Delta x + c_1^{min,t-1} }\right)}{\left(\dfrac{x_2 -y_2 }{\Delta y - \Delta x + c_1^{min,t-1}}\right)} \right) \\
      &=w_1^L(c_1^{max,t-1}-\tilde{c}_1)\left( \dfrac{\tilde{c}_1 -c_1^{min,t-1}}{\alpha + \tilde{c}_1 } \right)
  \end{align*}
  Setting its derivative to zero, its stationary point is $\tilde{c}_1= -\alpha + \sqrt{(\alpha+c_1^{min,t-1})(\alpha+c_1^{max,t-1})}$. In another case, if there is at least one successful manipulation, $\text{longEU}^+(1,2,(\tilde{c}_1,0) | W^t)=w_1^L(c_1^{max,t-1}-\tilde{c}_1)\left( \dfrac{\tilde{c}_1 -c_1^{min,t-1}}{\alpha + \tilde{c}_1 } \right)\dfrac{\alpha + h_1((w_1^{max,t-1},w_2^{min,t-1}))}{h_1((w_1^{max,t-1},w_2^{min,t-1}))-c_1^{min,t-1}}$, then we have the same stationary point. In addition, if the longEU value from a second objective is higher, the leader will pay the cost of $(0,\tilde{c}_2)$ instead. Similarly, the stationary point is $\tilde{c}_2= \alpha +  \sqrt{(-\alpha+c_2^{min,t-1})(-\alpha+c_2^{max,t-1})}$. The rest of proof in this case is similar to Case (I).
    
    \textbf{Case (III)} For $c^o_1 \leq 0, c^o_2 > 0$: The proof is similar to Case (II).\\

    \noindent \textbf{General Case}: \\
    For the case where $|\mathcal{A}^L|\geq1$ and $|\mathcal{A}^F|\geq 2$, at $t=0$ for $d=1,2$, set $c^{max,t}_d=\min_{(l,f),f\neq BR(l)} c_d^{max,t}(l,f)$ as the global minimal single-objective cost providing the current best utility $u^{t}_{best}$ . Following the previous proof and updating all parameters based on the success of the manipulation, it yields that at round $t\geq 1$, the optimised cost of longEU$^+$ function is either
    \begin{itemize}
        \item $(\tilde{c}_1(l,f),0)$ such that $\tilde{c}_1= -\alpha(l,f)+\sqrt{(\alpha(l,f)+c_1^{min,t-1}(l,f))(\alpha(l,f)+c_1^{max,t-1})}$ where $p^t((\tilde{c}_1(l,f),0))<1$ or
        \item $(0,\tilde{c}_2(l,f))$ such that $\tilde{c}_2= \alpha(l,f)+\sqrt{(- \alpha(l,f)+c_2^{min,t-1}(l,f))(-\alpha(l,f)+c_2^{max,t-1})}$ where $p^t((0,\tilde{c}_2(l,f)))<1$  or
        \item $(c_1^s(l,f),0)$ such that $c_1^s(l,f):=  \min \{ c_1(l,f) < c^{max,t-1}_1 \ \text{where} \ p^t((c_1(l,f),0))=1 \}$ or
        \item $(0,c_2^s(l,f))$ such that $c_2^s(l,f):=  \min \{ c_2(l,f) < c^{max,t-1}_2 \ \text{where} \ p^t((c_2(l,f),0))=1 \}$ for some $(l,f)$
    \end{itemize} 
    Note that $\alpha(l,f)$ is a constant based on a payoff matrix and a manipulated action. The last two options come from the case when a successful manipulation creates a new utility which is better than the current best cost $(c^{max,t}_1,0)$ or $(0,c^{max,t}_2)$ with the probability of 1. If such a situation occurs, the longEU$^+$ function restricted to the single objective becomes 
    \begin{align*}
        \text{longEU}^+(l,f,\tilde{c}_d(l,f) | W^t) & = w_d^L(c_d^{max,t-1}(l,f)-\tilde{c}_d(l,f)) p^t((\tilde{c}_d(l,f),0)) & \text{if} \ \tilde{c}_d(l,f)\leq c_d^s(l,f) \\    
        & \text{or} \ w_d^L(c_d^{max,t-1}(l,f)-\tilde{c}_d(l,f)) & \text{if} \ \tilde{c}_d(l,f) > c_d^s(l,f). 
    \end{align*}  
    Thus the leader will pick one of the cost options as shown. However, when one of the last two options is picked, the current best cost is updated but the estimated weight has no update. The leader then picks the cost from the first two options again and the convergence is still guaranteed.
\end{proof}

For Algorithm \ref{alg:longEU_policy_RWMC}, when $T \rightarrow \infty$, the algorithm will manipulate if the best minimal cost $\bm{c}_{r}$ can provide a better utility over the current one. Intuitively, if it keeps randomising a new weight infinitely, the optimal cost will be guaranteed. However, the process can be very slow compared to Algorithm \ref{alg:longEU_policy_PFR} since there is a risk that Algorithm \ref{alg:longEU_policy_RWMC} decides to not manipulate and the feasible weight region has no update. For instance, if we have no preference information and the leader has a small area of acceptable cost then it has a very low chance to find the initial acceptable cost. \\

\subsection{Proof of Theorem \ref{thm:alg2_converge}}

\begin{proof}
    When $T \rightarrow \infty$, Algorithm \ref{alg:longEU_policy_RWMC} has a slight modification: Once the minimal cost $\hat{\bm{c}}_{r}(l,f)$ is identified, the algorithm finds the best action by $(l_r,f_r)= \underset{l,f,\bm{c}}{\mathrm{argmax}} \ \text{longEU}^+(l,f,\hat{\bm{c}}_{r}(l,f) | t,T, \hat{\bm{w}}_{r})$ and checks whether $\text{longEU}^+(l_r,f_r,\hat{\bm{c}}_{r}(l_r,f_r) | t,T, \hat{\bm{w}}_{r}) < 0$ to explore a new cost; otherwise, play the current best cost. \\

    By this process, the estimated weight may not be updated if the algorithm does not explore a new cost. Thus, we still need to prove that the probability of exploring a new cost is non-zero. Since all best response actions $BR(l)$ are known for all $l \in \mathcal{A}^L$, then the play-safe cost for leader's action $l$ and follower's action $f$, $\bm{c}^o(l,f)=u^F(l,BR(l))-u^F(l,f)$. Let $c_d^{min,t}(l,f) \geq 0$ be the lower bound of the $d^{th}$-objective cost that is possible to manipulate the follower to choose action $(l,f)$ at round $t$ and is updated based on the boundary of a feasible region. Note that $\{c_d^{min,t}(l,f)\}_t$ is an increasing sequence since the size of the feasible region shrinks when an informative cost is offered. 
    
    Let $c_d^{max,t}(l,f)$ be the $d^{th}$-objective cost that provides the current best utility $u^t_{best}$ for action $(l,f)$ at round $t$..  
    By Lemma \ref{thm:better_cost_longEU}, $\{c_d^{max,t}(l,f)\}_t$ is a decreasing sequence. From Lemma \ref{thm:optimal_cost_C1}, the optimal cost $\bm{c}^*$ is a single-objective cost. Then there exists $l^*,f^*,d^*$ such that $c_{d^*}^{min,t}(l^*,f^*)<c^*_{d^*}<c_{d^*}^{max,t}(l^*,f^*)$ for all $t$. 
    In addition, Lemma \ref{thm:optimal_cost_C1}, fixing action $(l,f)$, the optimal cost with respect to the weight is always a single-objective cost. Define a function $f_{(l,f,d)}(\bm{w}) =\dfrac{\bm{w} \cdot \bm{c}^o(l,f)}{w_d}$

      which returns the minimal $d^{th}$-objective cost to manipulate the follower at action $(l,f)$ with respect to weight $\bm{w} \in  FR(\bm{w}; \emptyset)$. {Since $f_{(l,f,d)}(\bm{w})$ is surjective}, for all $c^t_d$ such that $c_d^{min,0}(l,f)<c^t_d<c_d^{max,0}(l,f)$, there is a weight requiring such a cost to manipulate. In other words, by assigning a uniform distribution over the current feasible region, the probability of exploring any new cost $c^t_d$ is non-zero. As long as, there is an update from any manipulation, the  $[c_{d^*}^{min,t}(l^*,f^*),c_{d^*}^{max,t}(l^*,f^*)]$ is also updated (shrinking).  Therefore $ c_{d^*}^{max,t}(l^*,f^*)-c_{d^*}^{min,t}(l^*,f^*) \rightarrow 0$. 
\end{proof}

\section{Additional Empirical Results} \label{sec:appendix_results}
Note that the code is available on  \href{https://github.com/S-Phurinut/MO-STB-2p-game}{github.com/S-Phurinut/MO-STB-2p-game}. 
\subsection{Uniformly random game with linear utility} \label{sec:appendix_results_linear}
 In this experiment, the payoff matrices as well as the utility weights are uniformly randomised. The payoff value for each objective is bounded between 0 and 1. We found that only approximately 30-40$\%$ of 2D and 3D random games are beneficial games; but 40-50$\%$ for 10D random games. The average cumulative regret $\pm$ standard error over replications for different maximum numbers of rounds  $T$ are reported (note that longEU uses knowledge about $T$ and thus we here report only on the final cumulative regret given $T$). With a larger number of available pairs of actions (2x2 actions), 10D games show a smaller cumulative regret since they have a higher percentage of beneficial games compared to 2D and 3D games (See Figures \ref{Appendix:fig:2D-uniform-game:cumregret}, \ref{Appendix:fig:3D-uniform-game:cumregret} and \ref{Appendix:fig:10D-uniform-game:cumregret}). 

 Additionally, we report the final cumulative regret plot of 2D games with a large action space, i.e., 10x10 actions. Figure \ref{Appendix:fig:2D-uniform-game:cumregret-10x10} shows a similar trend to the case with a smaller action space. However, the cumulative regret of the EU policy is much worse due to the large set of pairs of the manipulated action space.  

\subsection{Uniformly random game with Cobb-Douglas utility} \label{sec:appendix_results_Cobb}

Generally, EU and longEU policies can be applied directly to any form of parametric utility function. The feasible region of parameters is shaped according to the form of the utility function to approximate the probability of accepting the offer. 

In this section, we test our proposed policy with a non-linear utility function. We use a common non-linear utility function in the economic literature, the Cobb-Douglas (CD) function. 
The general form of a CD  function for a set of $n$ inputs is
$$ u(x_1,x_2,...,x_d;\bm{w})= \prod^D_{i=1} x_i^{w_i}$$
where $w_i$ are parameters determining the relative importance of each objective. One can normalise the Cobb-Douglas function by setting $\sum_{i=1}^D w_i=1$. 

We run similar experiments as in the previous section with 2D games, 2x2 actions and normalised Cobb-Douglas utility for both players, where $w_i$ are uniformly randomised over the simplex space for both players. We only compare the performance of the PFR approach with different models to approximate a probability of accepting the offer based on \eqref{eq:prob_accept_PFR}; 1. Linear utility 2. Cobb-Douglas (CD) utility. Figure \ref{Appendix:fig:2D-uniform-game:cumregret-CD} depicts a small gap in cumulative regret, even if the leader uses an inconsistent model to approximate the follower's preference.

 \begin{figure}[h]
        \begin{subfigure}{\textwidth}
         \centering
         \vspace{-1em}
         \includegraphics[width=0.7\textwidth]{pic/random_label.pdf}
     \end{subfigure}
     \begin{subfigure}{0.48\textwidth}
         \centering
         \includegraphics[width=0.8\textwidth]{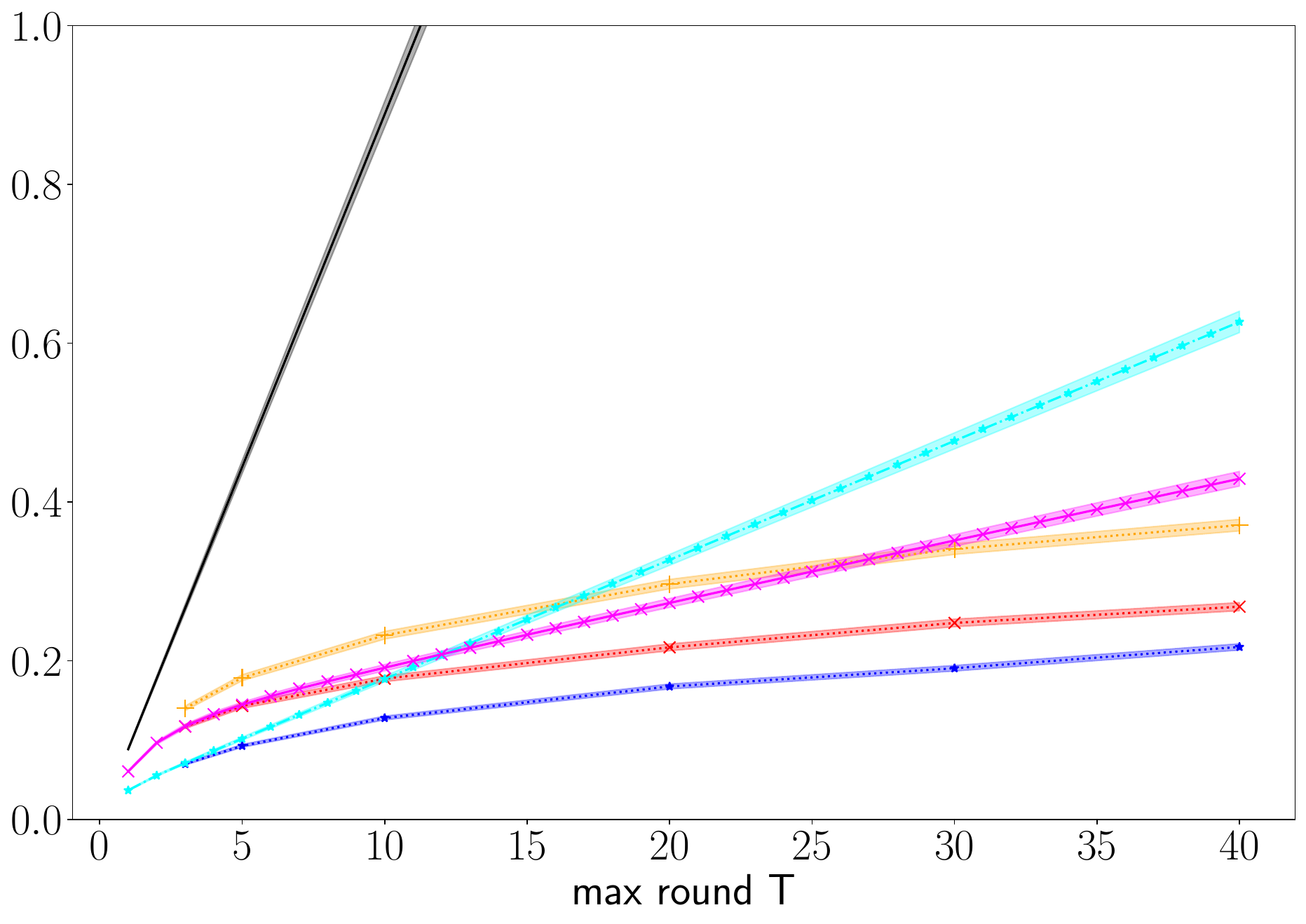}
         \caption{1x2 actions, constraint (C1)}
         \vspace{1em}
     \end{subfigure}
    \hfill 
     \begin{subfigure}{0.48\textwidth}
         \centering
         \includegraphics[width=0.8\textwidth]{pic/2D_uniform_game-full-cumregret_1by2action_Bcost.pdf}
         \caption{1x2 actions, constraint (C2) }
         \vspace{1em}
     \end{subfigure} 
     
     \begin{subfigure}{0.48\textwidth}
         \centering
         \includegraphics[width=0.8\textwidth]{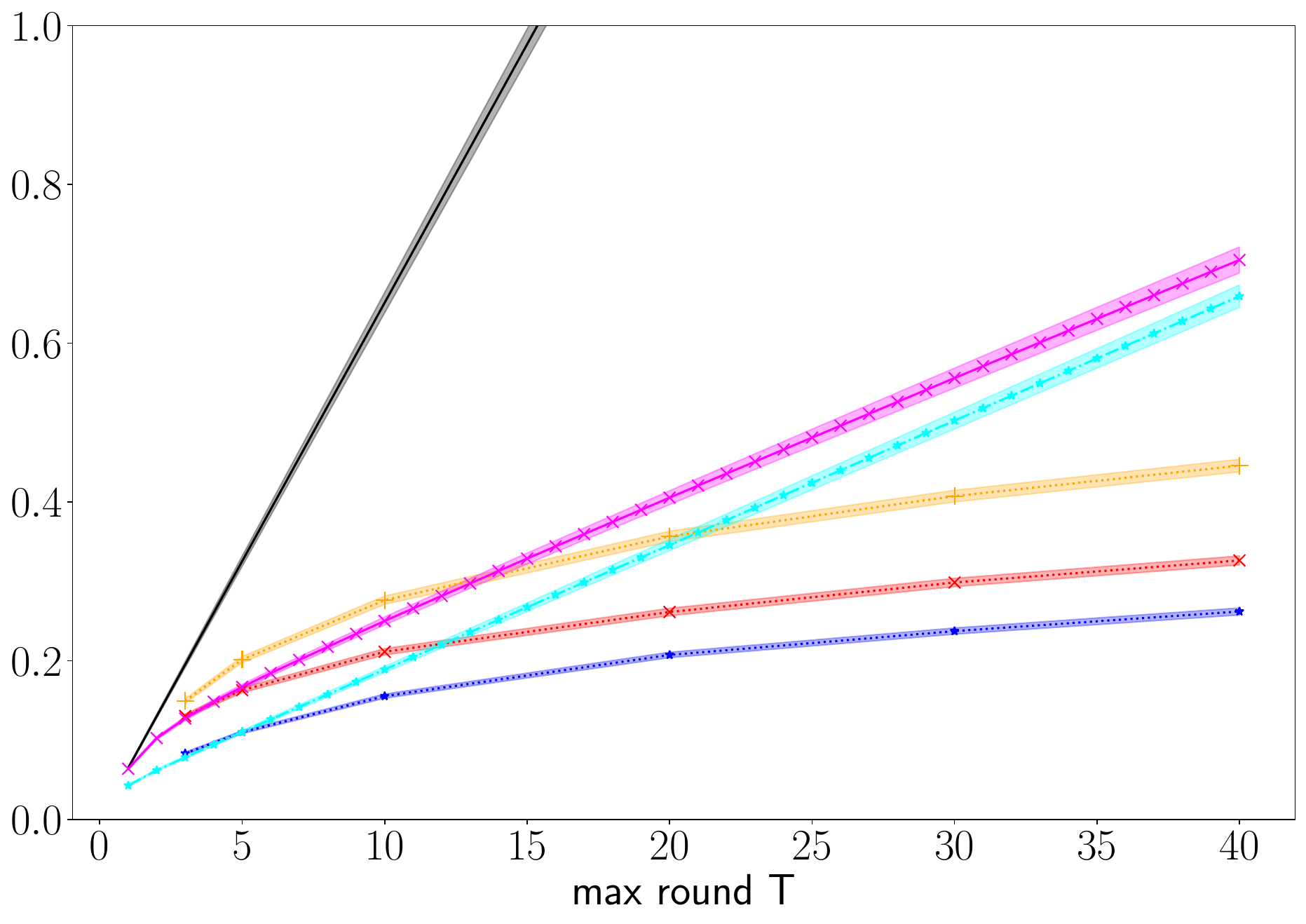}
         \caption{2x2 actions, constraint (C1) }
     \end{subfigure}
     \hfill   
     \begin{subfigure}{0.48\textwidth}
         \centering
         \includegraphics[width=0.8\textwidth]{pic/2D_uniform_game-full-cumregret_2by2action_Bcost.pdf}
         \caption{2x2 actions, constraint (C2) }
     \end{subfigure}
    \vspace{1em}
    \caption{Final cumulative regret for 2D uniformly random games from 10000 replications with linear utility function}
    \label{Appendix:fig:2D-uniform-game:cumregret}
    \vspace{1.5em}
\end{figure}

\begin{figure}[h]
     \begin{subfigure}{0.48\textwidth}
         \centering
         \includegraphics[width=0.8\textwidth]{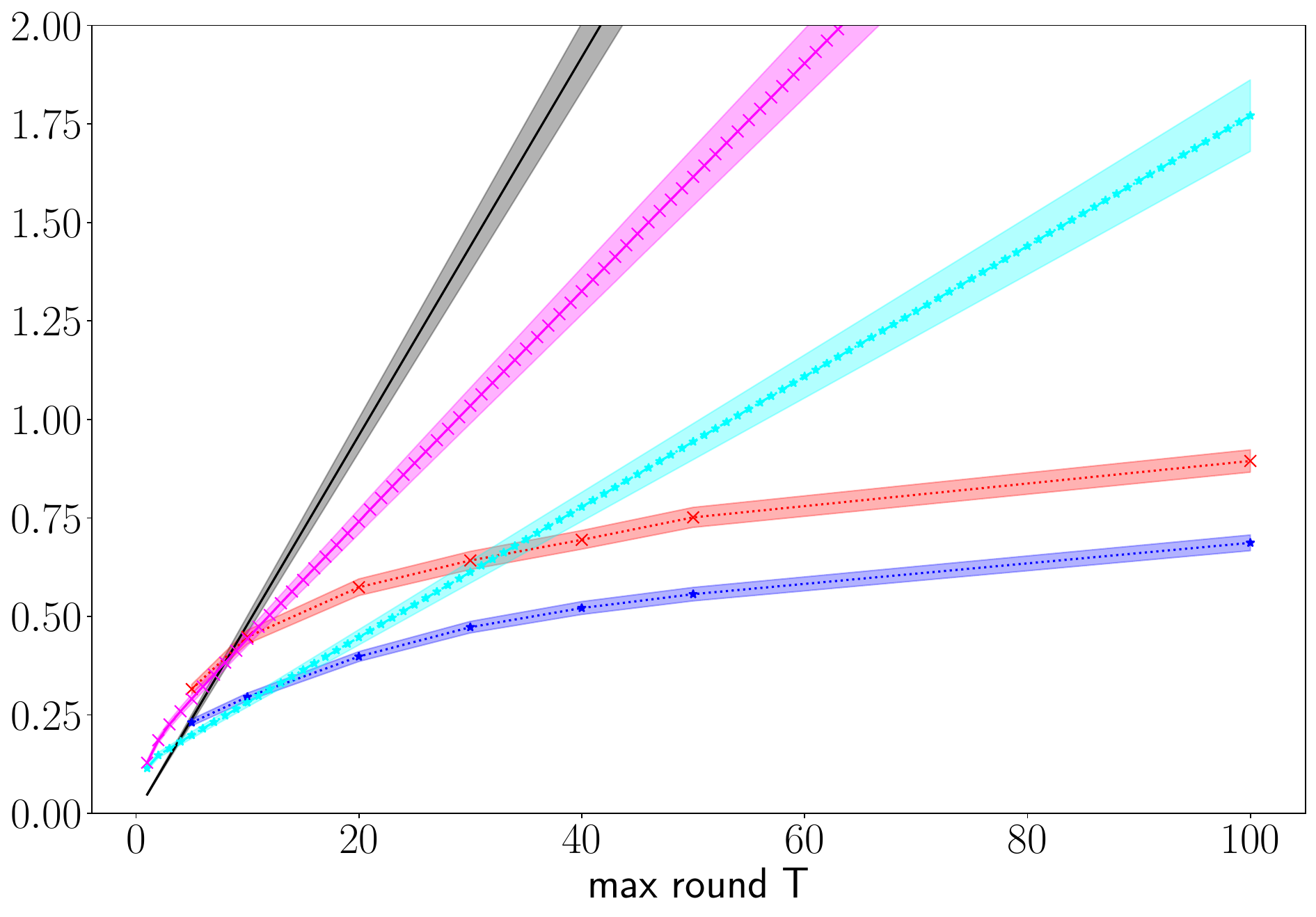}
         \caption{10x10 actions, constraint (C1)}
         \vspace{1em}
     \end{subfigure}
    \hfill 
     \begin{subfigure}{0.48\textwidth}
         \centering
         \includegraphics[width=0.8\textwidth]{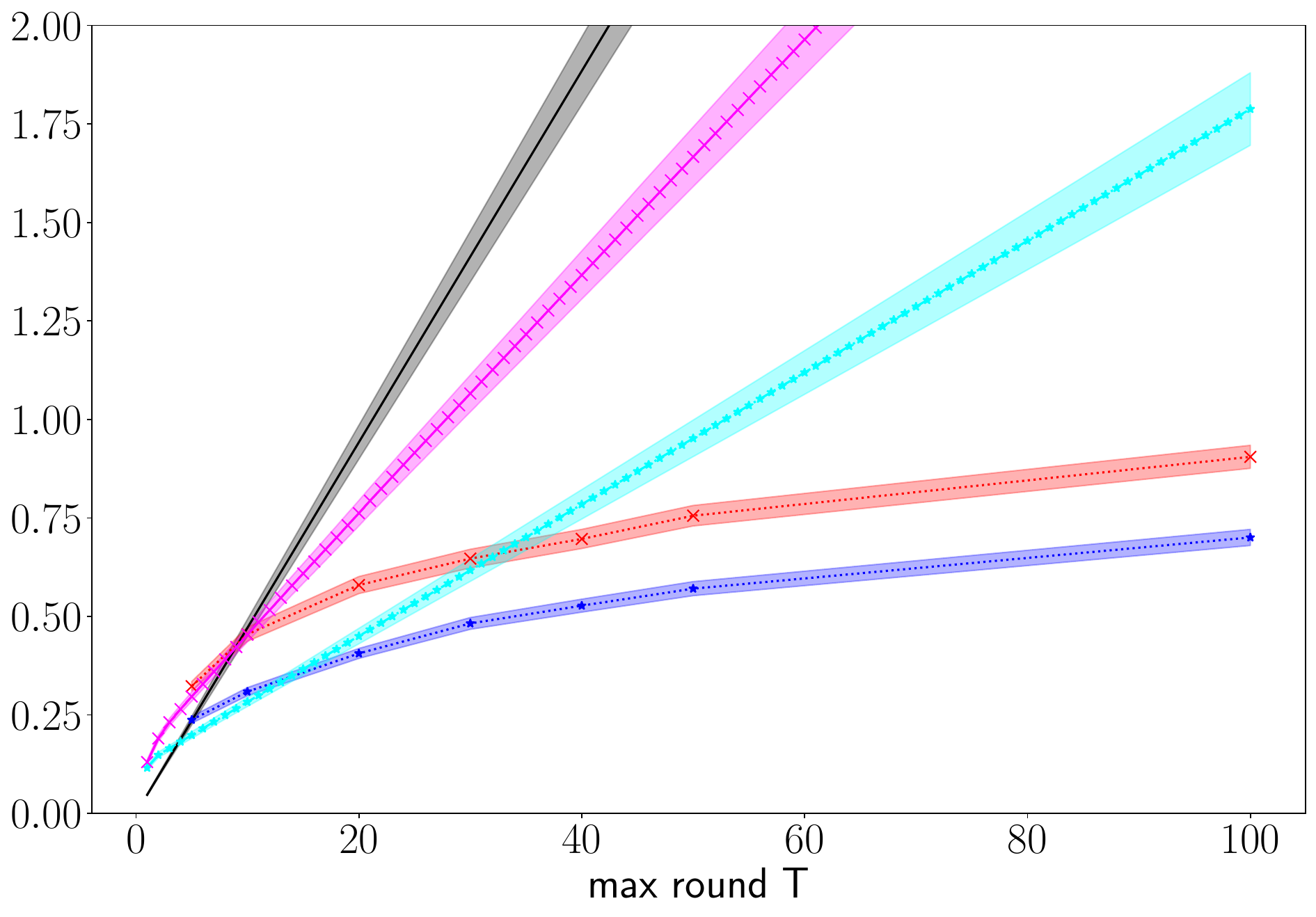}
         \caption{10x10 actions, constraint (C2) }
         \vspace{1em}
     \end{subfigure}
     \vspace{1em}
    \caption{Final cumulative regret for 2D uniformly random games from 1000 replications with linear utility function}
    \label{Appendix:fig:2D-uniform-game:cumregret-10x10}
    \vspace{1.5em}
\end{figure}
     
\begin{figure}[h]
     \begin{subfigure}{0.48\textwidth}
         \centering
         \includegraphics[width=0.8\textwidth]{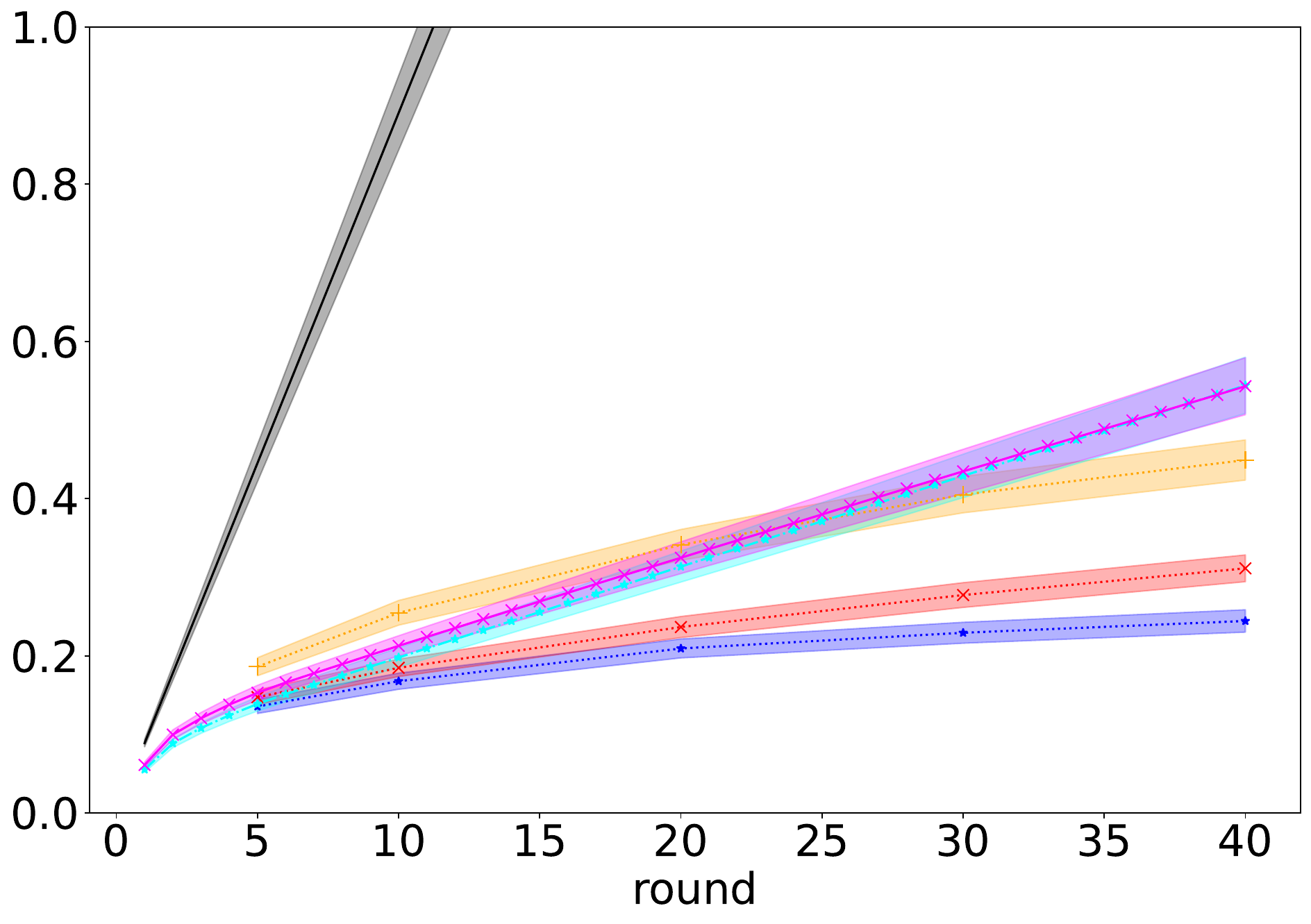}
         \caption{1x2 actions, constraint (C1)}
         \vspace{1em}
     \end{subfigure}
    \hfill 
     \begin{subfigure}{0.48\textwidth}
         \centering
         \includegraphics[width=0.8\textwidth]{pic/3D_uniform_game-full-cumregret_1by2action_Bcost.pdf}
         \caption{1x2 actions, constraint (C2) }
         \vspace{1em}
     \end{subfigure}
     \vspace{1em}
     \begin{subfigure}{0.48\textwidth}
         \centering
         \includegraphics[width=0.8\textwidth]{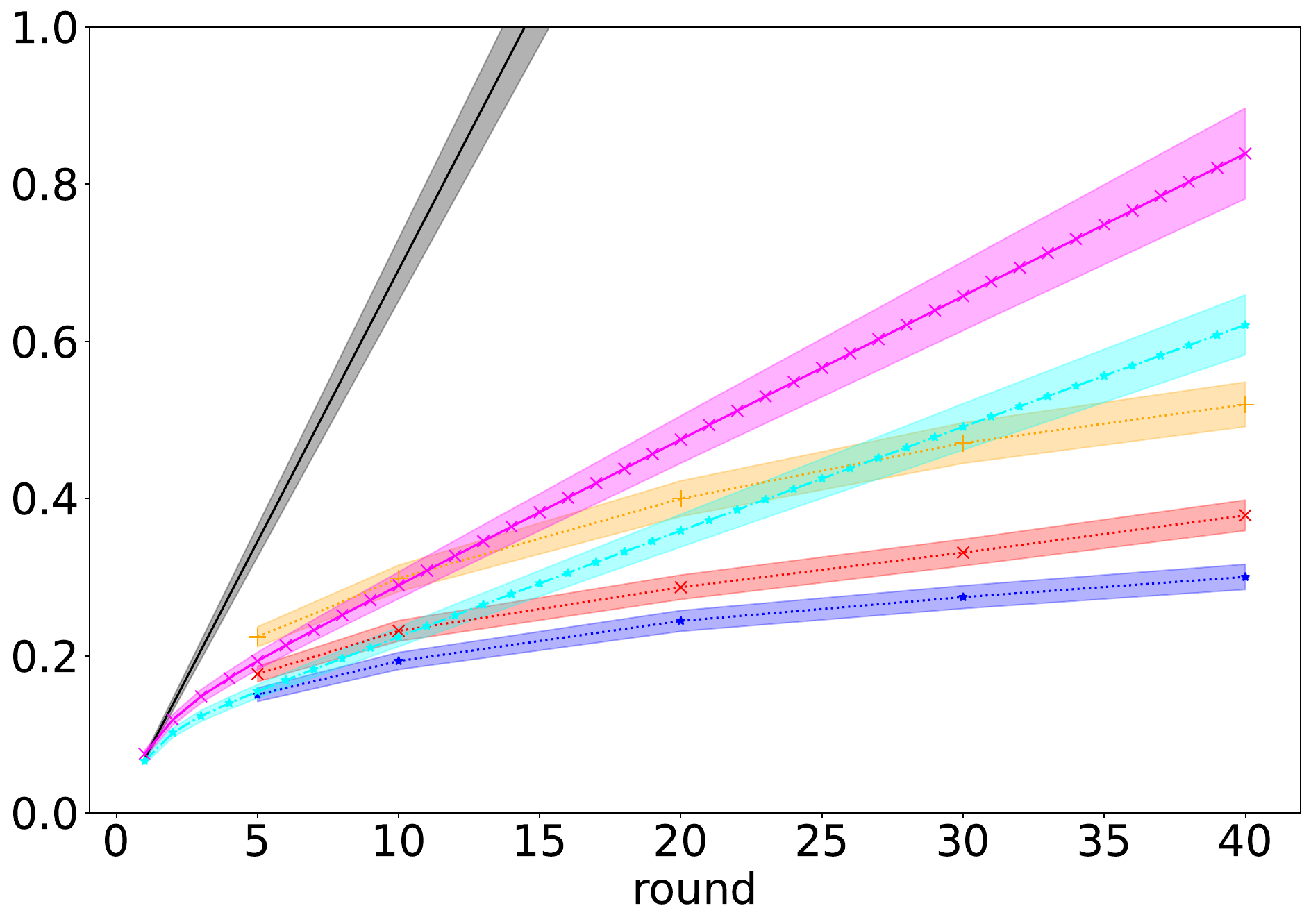}
         \caption{2x2 actions, constraint (C1) }
     \end{subfigure}
     \hfill   
     \begin{subfigure}{0.48\textwidth}
         \centering
         \includegraphics[width=0.8\textwidth]{pic/3D_uniform_game-full-cumregret_2by2action_Bcost.pdf}
         \caption{2x2 actions, constraint (C2) }
     \end{subfigure}
     \vspace{0.5em}
    \caption{Final cumulative regret for 3D uniformly random games from 1000 replications with linear utility function}
    \label{Appendix:fig:3D-uniform-game:cumregret}
    \vspace{1em}
\end{figure}

\begin{figure}[h]
     \begin{subfigure}{0.48\textwidth}
         \centering
\includegraphics[width=0.8\textwidth]{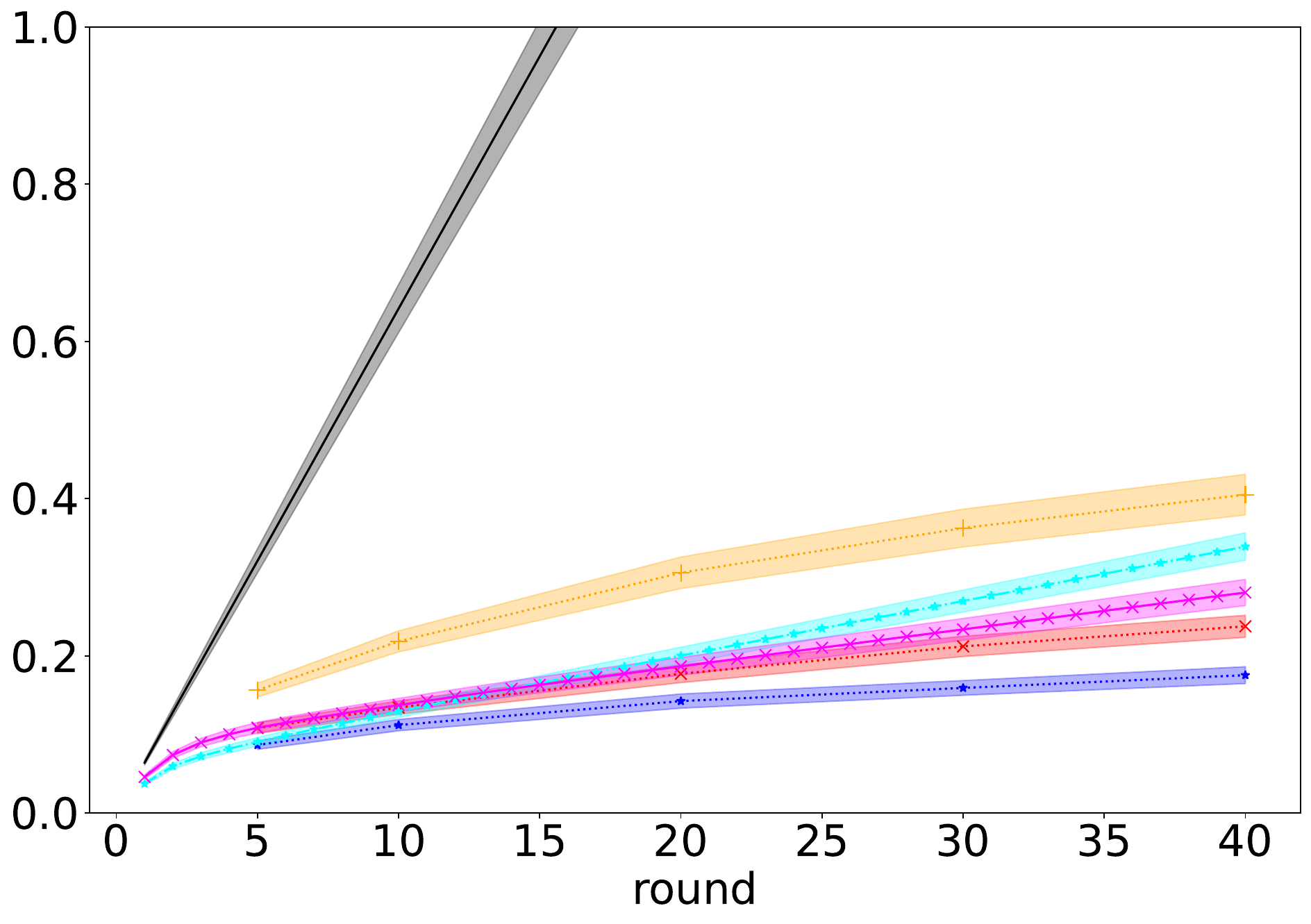}
         \caption{1x2 actions, constraint (C1)}
     \end{subfigure}
    \hfill 
     \begin{subfigure}{0.48\textwidth}
         \centering
         \includegraphics[width=0.8\textwidth]{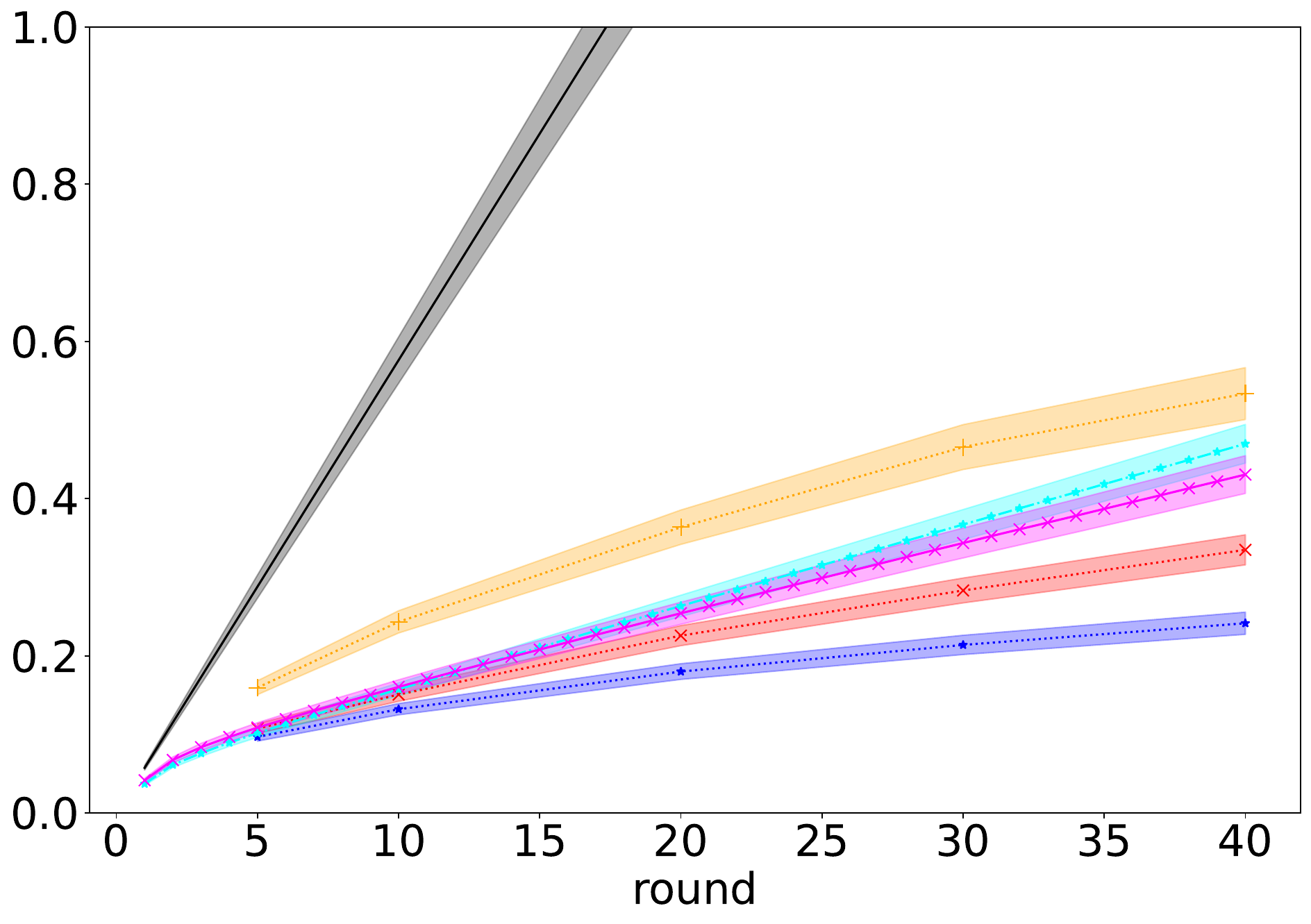}
         \caption{1x2 actions, constraint (C2) }
     \end{subfigure}
     
     \vspace{1em}
     \begin{subfigure}{0.48\textwidth}
         \centering
         \includegraphics[width=0.8\textwidth]{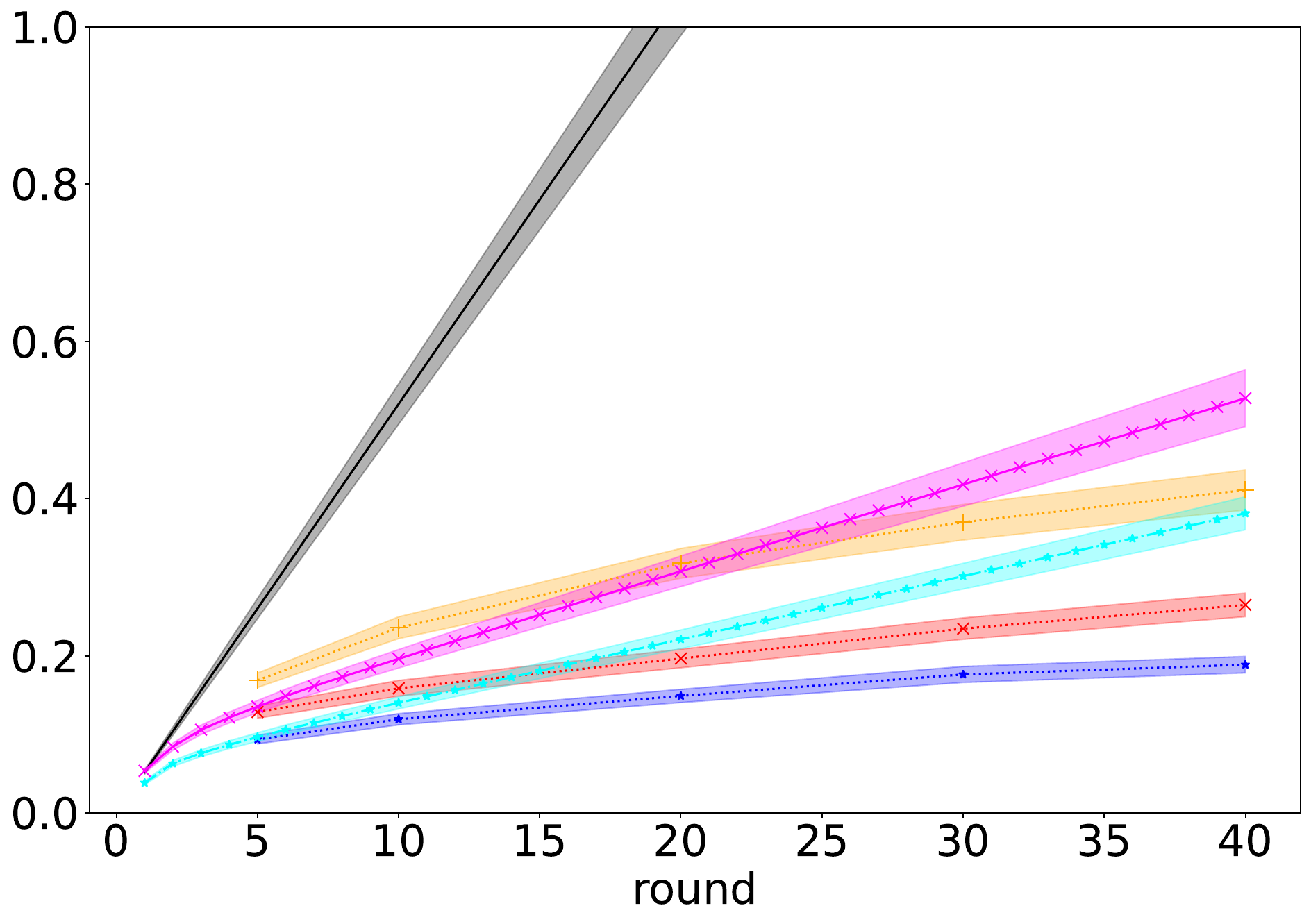}
         \caption{2x2 actions, constraint (C1) }
         \label{fig:10D-uniform-game:cumregret-2x2action-C1}
     \end{subfigure}
     \hfill   
     \begin{subfigure}{0.48\textwidth}
         \centering
         \includegraphics[width=0.8\textwidth]{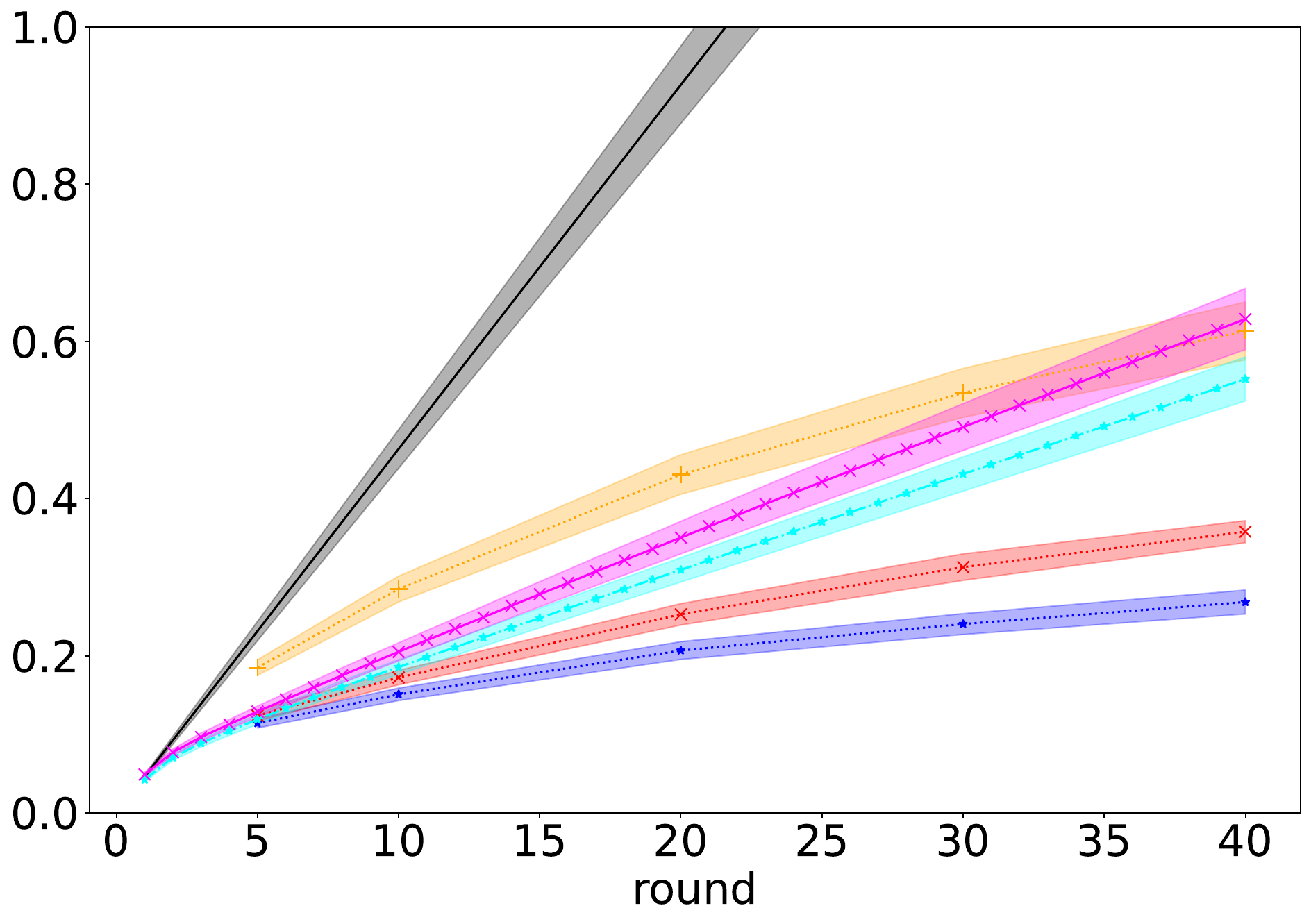}
         \caption{2x2 actions, constraint (C2) }
     \end{subfigure}
     \vspace{0.7em}
    \caption{Final cumulative regret for 10D uniformly random games from 1000 replications with linear utility function}
    \label{Appendix:fig:10D-uniform-game:cumregret}
    \vspace{2em}
\end{figure}

 \begin{figure}[h]
        \begin{subfigure}{\textwidth}
         \centering
         \vspace{-1em}
         \includegraphics[width=0.7\textwidth]{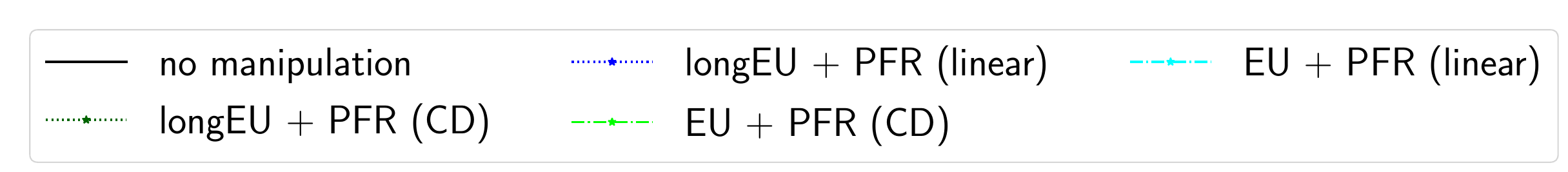}
     \end{subfigure}
     \begin{subfigure}{0.48\textwidth}
         \centering
         \includegraphics[width=0.8\textwidth]{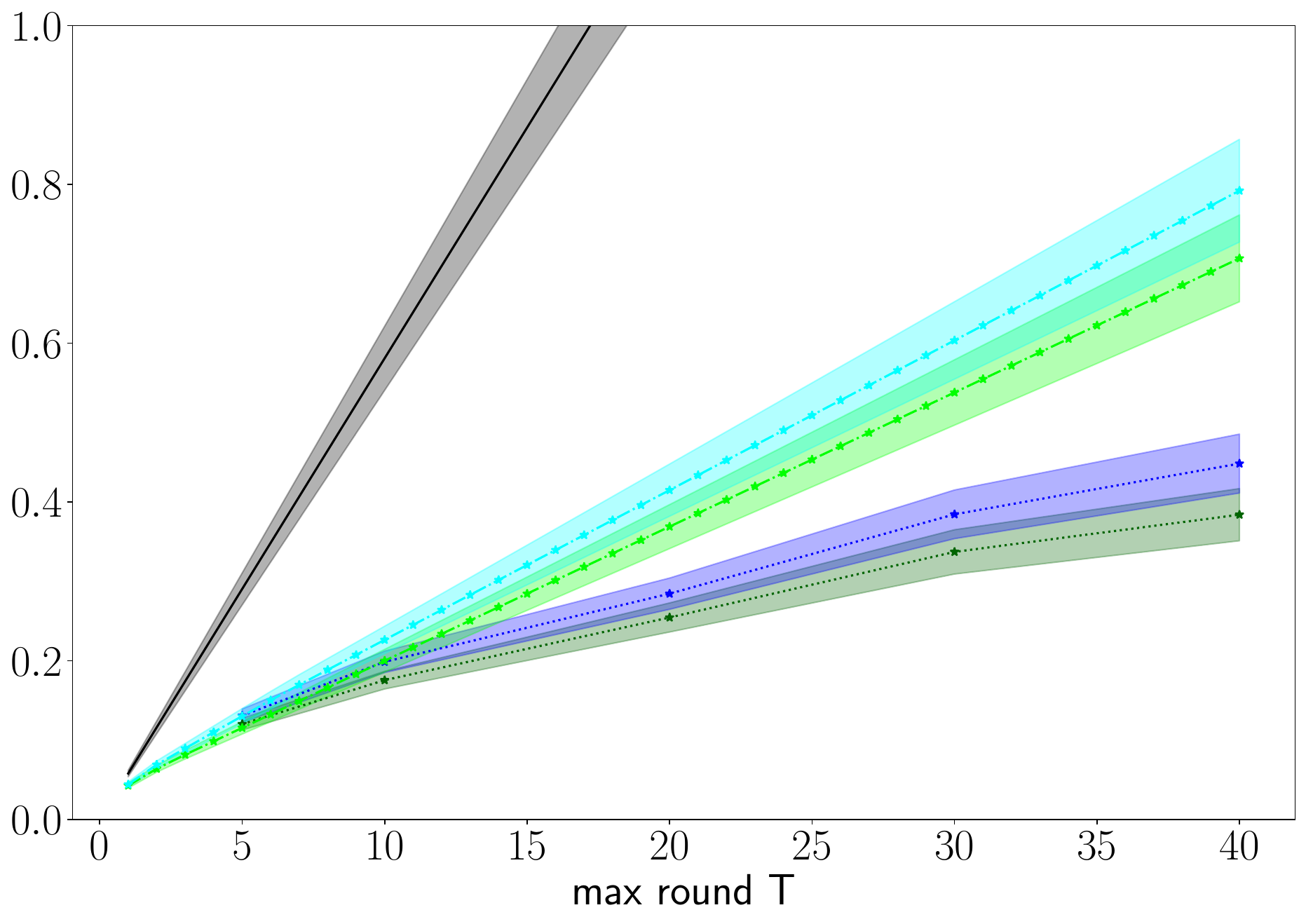}
         \caption{2x2 actions, constraint (C1)}
         \vspace{1em}
     \end{subfigure}
    \hfill 
     \begin{subfigure}{0.48\textwidth}
         \centering
         \includegraphics[width=0.8\textwidth]{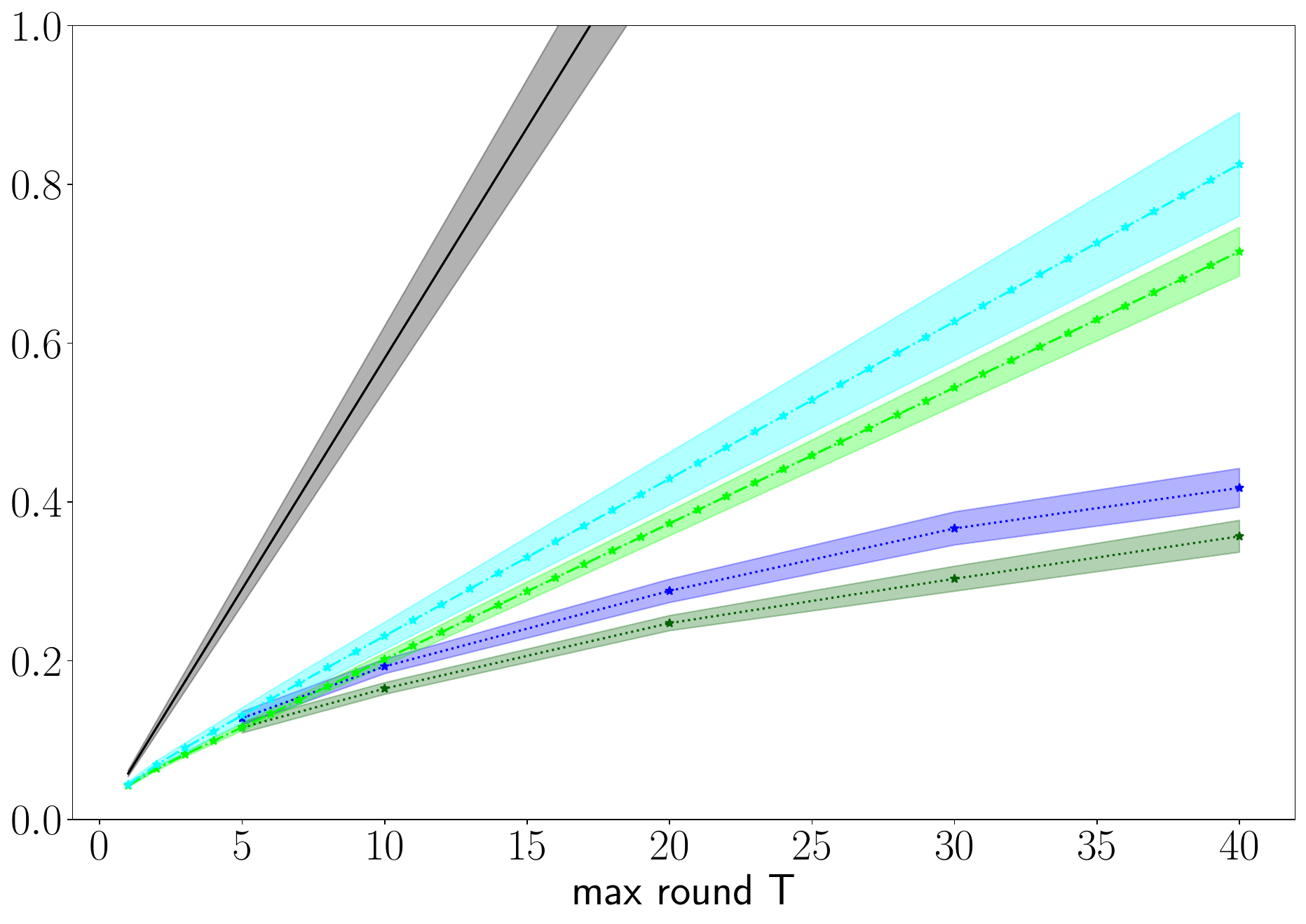}
         \caption{2x2 actions, constraint (C2) }
         \vspace{1em}
     \end{subfigure} 
    \vspace{1em}
    \caption{Final cumulative regret for 2D uniformly random games from 1000 replications with Cobb-Douglas utility function}
    \label{Appendix:fig:2D-uniform-game:cumregret-CD}
\end{figure}
\end{document}